\documentclass[conference]{IEEEtran}
\ifCLASSINFOpdf
\else
\fi
\usepackage{hyperref}

\usepackage{cite}
\usepackage{xcolor}
\usepackage[utf8]{inputenc}

\usepackage{subcaption}

\definecolor{forestgreen(web)}{rgb}{0.13, 0.55, 0.13}

\newcommand\update[2]{#2}

\usepackage{makecell}
\usepackage{ragged2e}  
\usepackage{xspace}
\usepackage{multirow}
\usepackage{booktabs}
\newcommand{\sr}{\rule[-0.13cm]{0pt}{0.45cm}}

\renewcommand{\labelenumi}{\arabic{enumi}.}

\usepackage{amsmath,amssymb,amsfonts}
\usepackage{amsthm}
\newtheorem{prop}{Proposition}
\newtheorem{lemma}{Lemma}
\newtheorem{theorem}{Theorem}

\usepackage[ruled,noend,linesnumbered]{algorithm2e}
\usepackage{algpseudocode}
\usepackage[scr=boondox]{mathalfa}
\newcommand{\col}[1][]{\mathscr{c \mkern -3mu o \mkern -3mu l}\mkern -3mu_#1}
\newcommand{\colprime}[1][]{\mathscr{c \mkern -3mu o \mkern -3mu l'}\mkern -3mu_#1}
\newcommand{\Listone}{\mathcal{L}}
\newcommand{\Listtwo}{\mathcal{P}}

\newcommand{\tx}[1][]{\tau_#1}

\newcommand{\state}[1][]{\mathcal{S}_#1}
\makeatletter 
\g@addto@macro{\@algocf@init}{\SetKwInOut{executed}{Actor }} 
\makeatother

\SetAlCapNameFnt{\footnotesize}
\SetAlCapFnt{\footnotesize}

\newcommand{\sysname}{Snappy\xspace}
\usepackage{pifont} 

\usepackage{tikz}
\newcommand*\circled[1]{\tikz[baseline=(char.base)]{
		\node[shape=circle,draw,inner sep=1pt] (char) {#1};}}

\newcommand{\bp}{\mathsf{bp}}
\newcommand{\Gr}{\mathbb{G}}
\newcommand{\Hr}{\mathbb{H}}
\newcommand{\Tr}{\mathbb{G}_T}
\newcommand{\Z}{\mathbb{Z}}

\usepackage{enumitem}
\setlist{nosep}

\usepackage[font={small}]{caption}
\newcommand{\mypara}[1]{\smallbreak \noindent \textbf{#1}}
\newcommand{\figsaver}{\vspace{-4pt}}
\newcommand{\tabsaver}{\vspace{-4pt}}
\newcommand{\floattabsaver}{\vspace{-15pt}}

\usepackage{mathtools}
\DeclarePairedDelimiter{\ceil}{\lceil}{\rceil}

\begin{document}
\title{Snappy: Fast On-chain Payments\\ with Practical Collaterals}
		

\author{\IEEEauthorblockN{Vasilios Mavroudis}
	\IEEEauthorblockA{University College London\\
		{\small v.mavroudis@ucl.ac.uk}}
	\and
\IEEEauthorblockN{Karl Wüst}
\IEEEauthorblockA{ETH Zurich\\
	{\small karl.wuest@inf.ethz.ch}}
	\and
\IEEEauthorblockN{Aritra Dhar}
\IEEEauthorblockA{ETH Zurich\\
	{\small aritra.dhar@inf.ethz.ch}}
	\and
\IEEEauthorblockN{Kari Kostiainen}
\IEEEauthorblockA{ETH Zurich\\
	{\small kari.kostiainen@inf.ethz.ch}}
\and
\IEEEauthorblockN{Srdjan Capkun}
\IEEEauthorblockA{ETH Zurich\\
	{\small srdjan.capkun@inf.ethz.ch}}}

\IEEEoverridecommandlockouts
\makeatletter\def\@IEEEpubidpullup{6.5\baselineskip}\makeatother
\IEEEpubid{\parbox{\columnwidth}{
		\textbf{\emph{Pre-print version of a conference publication to appear in:}} \\
		Network and Distributed Systems Security (NDSS) Symposium 2020\\
		23-26 February 2020, San Diego, CA, USA\\
	}
	\hspace{\columnsep}\makebox[\columnwidth]{}}

\maketitle

\begin{abstract}
	Permissionless blockchains offer many advantages but also have significant limitations including high latency. This prevents their use in important scenarios such as retail payments, where merchants should approve payments fast. Prior works have attempted to mitigate this problem by moving transactions off the chain. However, such Layer-2 solutions have their own problems: payment channels require a separate deposit towards each merchant and thus significant locked-in funds from customers; 
	payment hubs require very large operator deposits that depend on the number of customers; and side-chains require trusted validators. 
	
	In this paper, we propose \sysname, a novel solution that enables recipients, like merchants, to safely accept fast payments. In \sysname, all payments are on the chain, while small customer collaterals and moderate merchant collaterals act as payment guarantees. Besides receiving payments, merchants also act as statekeepers who collectively track and approve incoming payments using majority voting. In case of a double-spending attack, the victim merchant can recover lost funds either from the collateral of the malicious customer or a colluding statekeeper (merchant). \sysname overcomes the main problems of previous solutions: a single customer collateral can be used to shop with many merchants; merchant collaterals are independent of the number of customers; and validators do not have to be trusted. 
	Our Ethereum prototype shows that safe, fast ($<$2 seconds) and cheap payments are possible on existing blockchains.
\end{abstract}

\section{Introduction}\label{sec:intro}

Cryptocurrencies based on permissionless blockchains have shown great potential in decentralizing the global financial system, reducing trust assumptions, increasing operational transparency and improving user privacy. However, this technology still has significant limitations, preventing it from posing as valid alternative to established transaction means such as card payments. 

One of the main problems of permissionless blockchains is \emph{high latency}. For example, in Ethereum~\cite{wood2014ethereum}, users have to wait approximately 3 minutes (10 blocks) before a new payment can be safely accepted~\cite{karame2015misbehavior,lei2015exploiting}. In comparison, traditional and centralized payment systems such as VISA can confirm payments within 2 seconds~\cite{visatime1,visatime2,emvtime1}. High latency makes permissionless blockchains unsuitable for many important applications such as point-of-sale purchases and retail payments.

To improve blockchain performance, various consensus schemes have been proposed~\cite{bano2017consensus}. While techniques like sharding and Proof-of-Stake can increase the \emph{throughput} of blockchains significantly, currently there are no promising solutions that would drastically decrease the latency of permissionless blockchains.

Thus, researchers have explored alternative avenues to enable fast payments over slow blockchains. Arguably, the most prominent approach are Layer-2 solutions that move transaction processing off the chain and use the blockchain only for dispute resolution and occasional synchronization. 
However, such solutions have their own shortcomings. For example, payment channels require a separate deposit for each channel, resulting in large locked-in funds for users such as retail customers~\cite{miller2017sprites,LindNEKPS18}. Payment networks cannot guarantee available paths and are incompatible with the unilateral nature of retail payments from customers to merchants \cite{poon2016bitcoin,network2018cheap}. Payment hubs~\cite{dziembowski2019perun,heilman2017tumblebit} require either a trusted operator or huge collaterals that are equal to the total expenditure of all customers~\cite{gudgeonsok,McCorryBBMM18}. 
Side-chains based on permissioned BFT consensus require 2/3 trusted validators and have high communication complexity~\cite{dilley2016strong,back2014enabling}.

\mypara{Our solution.}
In this paper, we introduce \sysname, a system that enables safe and fast (zero-confirmation) on-chain payments. \sysname can be used today on top of low-throughput and high-latency blockchains such as Ethereum and in the future on top of (sharded) high-throughput and mid-latency blockchains.

We tailor our solution for application scenarios such as retail payments, but emphasize that our design can be used in any scenario where a large number of users (e.g., 100,000 customers) make payments towards a moderate set of recipients (e.g., 100 merchants). In \sysname, the merchants form a joint consortium that may consist of large retailers with several stores globally or small local stores from the same neighborhood. The merchants need to communicate to be able accept fast-payments safely, but importantly neither the merchants nor the customers have to trust each other.

\sysname relies on customer collaterals that enable merchants to safely accept payments before the transaction has reached finality in the blockchain. The collaterals serve as payment guarantees and are deposited by customers to a smart contract during system enrollment. If the transaction corresponding to an accepted payment does not appear in the  blockchain within a reasonable time (double-spending attack), the victim merchant can recoup the lost funds from the malicious customer's collaterals. The customer's deposit should cover the value of the customer's purchases within the latency period of the blockchain (e.g., 3 minutes in Ethereum) which makes them small in practice (e.g., \$100 would suffice for many users). 
Moreover, customers do not need to repeatedly replenish their collaterals, as they are used only in the case of attack.

In \sysname, the payment recipients (merchants) act as \emph{untrusted statekeepers} whose task is to track and collectively approve incoming transactions. To initiate a fast payment, a customer creates a transaction that transfers funds to a smart contract and indicates the merchant as beneficiary. The recipient merchant sends the transaction to the statekeepers and proceeds with the sale only if a \emph{majority} of them approves it with their signatures. Upon collecting the required signatures, the merchant broadcasts the transaction in the blockchain network to be processed by an Arbiter smart contract. Once the transaction is processed and logged by the Arbiter, the payment value is forwarded to the merchant.

Statekeepers must also deposit collaterals that protect merchants from attacks where a malicious statekeeper colludes with a customer. In such a case, the victim merchant can use the statekeeper's approval signatures as evidence to claim any lost funds from the misbehaving statekeeper's collateral. The size of statekeeper collateral is proportional to the total amount of purchases that all participating merchants expect within the blockchain latency period. 
Crucially, the statekeeper collaterals are independent of the number of customers which allows the system to scale. The main security benefit of statekeeper collaterals is that they enable fast and safe payments without trusted parties.

\mypara{Main results.} We prove that a merchant who follows the \sysname protocol and accepts a fast payment once it has been approved by the majority of the statekeepers never loses funds regardless of any combination of customer and statekeeper collusion. We also demonstrate that \sysname is practical to deploy on top of existing blockchains by implementing it on Ethereum. The performance of our solution depends primarily on number of participating statekeepers (merchants). For example, assuming a deployment with 100 statekeepers, a payment can be approved in less than 200 ms with a processing cost of \$0.16 (169k Ethereum gas), which compares favorably to card payment fees. 

\sysname overcomes the main problems of Layer-2 solutions in application scenarios such as retail payments. In contrast to BFT side-chains that assume that $2/3$ honest validators and require multiple rounds of communication, \sysname requires no trusted validators and needs only one round of communication. Unlike payment channels, \sysname enables payments towards many merchants with a single and small customer deposit. In contrast to payment networks, \sysname payments can always reach the merchants, because there is no route depletion. And finally, the statekeeping collaterals are practical even for larger deployments, compared to those in payment hubs, as they are independent of the number of customers in the system. 

\mypara{Contributions.} This paper makes the following contributions:

\begin{itemize}
	\item \emph{Novel solution for fast payments.} We propose a system called \sysname that enables fast and secure payments on slow blockchains without trusted parties using moderately-sized and reusable collaterals that are practical for both customers and merchants.

	\item \emph{Security proof.} We prove that merchants are guaranteed to receive the full value of all accepted payments, in any possible combination of double spending by malicious customers and equivocation by colluding statekeepers.   

	\item \emph{Evaluation.} We implemented \sysname on Ethereum and show that payment processing is fast and cheap in practice.
\end{itemize}

\smallskip
This paper is organized as follows: Section~\ref{sec:problem} explains the problem of fast payments, Section~\ref{sec:overview} provides an overview of our solution and Section~\ref{sec:protocols} describes it in detail. Section~\ref{sec:analysis} provides security analysis and Section~\ref{sec:eval} further evaluation. Section~\ref{sec:discussion} is discussion, Section~\ref{sec:related_work} describes related work, and Section~\ref{sec:conclusion} concludes the paper.

\section{Problem Statement}
\label{sec:problem}

In this section we motivate our work, explain our assumptions, discuss the limitations of previous solutions, and specify requirements for our system.

\subsection{Motivation}
\label{subsec:motivation}

The currently popular permissionless blockchains (e.g., Bitcoin and Ethereum) rely on Proof-of-Work (PoW) consensus that has well-known limitations, including low throughput (7 transactions per second in Bitcoin), high latency (3 minutes in Ethereum), and excessive energy consumption (comparable to a small country~\cite{digiconomist}). Consequently, the research community has actively explored alternative permissionless consensus schemes. From many proposed schemes, two prominent approaches, \emph{Proof of Stake} and \emph{sharding}, stand out~\cite{bano2017consensus}. 

Proof of Stake (PoS) systems aim to minimize the energy waste by replacing the computationally-heavy puzzles of PoW with random leader election such that the leader selection probability is proportional to owned staked. While the current PoS proposals face various security and performance issues~\cite{bano2017consensus},
the concept has shown promise in mitigating the energy-consumption problem.

Sharding systems increase the blockchain's throughput by dividing the consensus participants into committees (\emph{shards}) that process distinct sets of transactions. Recent results reported significant throughput increases in the order of thousands of transactions per second~\cite{kokoris2018omniledger, zamanirapidchain}. Despite several remaining challenges, sharding shows great promise in improving blockchain throughput. 

Sharding and PoS can also address transaction latency. Recent works such as Omniledger~\cite{kokoris2018omniledger} and RapidChain~\cite{zamanirapidchain} use both techniques and report latencies from 9 to 63 seconds, \update{7}{assuming common trust models like honest majority or 1/3 Byzantine nodes}. However, such measurements are achieved in fast test networks (permissionless blockchains rely on slow peer-to-peer networks) and under favorable work loads (e.g., largely pre-sharded transactions). 

Our conclusion is that while permissionless blockchain throughput and energy efficiency are expected to improve in the near future, latency will most likely remain too high for various scenarios such as point-of-sale payments and retail shopping, where payment confirmation is needed within 1-2 seconds. Therefore, \emph{in this paper we focus on improving blockchain payment latency}. We consider related problems like limited blockchain throughput as orthogonal problems with known solutions. Appendix~\ref{app:consensus} provides further background on permissionless consensus.

\subsection{System Model and Assumptions}
\label{subsec:assumptions}

\mypara{Customers and merchants.} We focus on a setting where $n$ users send payments to $k$ recipients such that $n$ is large and $k$ is moderate. One example scenario is a set of $k=100$ small shops where $n=100,000$ customers purchase goods at. 
Another example is $k=100$ larger retail stores with $n=1$ million customers~\cite{tesco2016,tesco2018}. 

We consider merchants who accept \emph{no risk}, i.e., they hand the purchased products or services to the customers, only if they are guaranteed to receive the full value of their sale. Therefore, naive solutions such as accepting zero-confirmation transaction are not applicable~\cite{karame2012double}. The customers are assumed to register once to the system (similar to a credit card issuance processes) and then visit shops multiple times.

We assume that merchants have permanent and fairly reliable Internet connections. Customers are not expected to be online constantly or periodically (initial online registration is sufficient). At the time of shopping, customers and merchants can communicate over the Internet or using a local channel, such as a smartphone NFC or smart card APDU interface.

\mypara{Blockchain.} We assume a permissionless blockchain that has sufficient throughput, but high latency (see motivation). The blockchain supports smart contracts.
We use Ethereum as a reference platform throughout this work, but emphasize that our solution is compatible with most permissionless blockchains with smart contracts~\cite{hanke2018dfinity, wanchain2018, nxt2014}. To ensure compatibility with existing systems like Ethereum, we assume that smart contracts have access to the current state of the blockchain, but not to all the past transactions.

\mypara{Adversary.} 
\update{1}{The main goal of this paper is to enable secure and fast payments. Regarding payment security}, we consider a strong adversary who controls an arbitrary number of customers and all other merchants besides the target merchant who accepts a fast payment. The adversary also controls the network connections between customers and merchants but cannot launch network-level attacks such as node eclipsing~\cite{heilman2015eclipse}. The adversary cannot violate the consensus guarantees of the blockchain, prevent smart contract execution, or violate contract integrity.

\update{1}{For payment liveness, we additionally require that sufficiently many merchants are responsive (see Section~\ref{sec:liveness}).}

\subsection{Limitations of Known Solutions}
\label{subsec:limitations}

A prominent approach to enable fast payments on slow permissionless blockchains is so called Layer-2 solutions that move transaction processing off the chain and use the blockchain only in case of dispute resolution and occasional synchronization between the on-chain and off-chain states. Here, we outline the main limitations of such solutions. \update{}{Section~\ref{sec:related_work} provides more details on Layer-2 solutions and their limitations.}

\mypara{Payment channels} transfer funds between two parties. The security of such solutions is guaranteed by separate deposits that must cover periodic expenditure in each individual channel.
In our ``small shops'' example with $k=100$ merchants and an average customer expenditure of $e=\$10$, the customers would need to deposit combined \$1,000. In our ``large retailers'' example with $k=100$ merchants and expenditure of $e=\$250$, the total deposit is \$25,000. Payment channels require periodic deposit replenishment.

\mypara{Payment networks} address the above problem of having to set up many separate channels by using existing payment channels to find paths and route funds. Payments are possible only when the necessary links from the source (customer) to the destination (merchant) exist. However, payment networks are unreliable, as guaranteeing the suitable links between all parties is proven difficult~\cite{prihodko2016flare}. Moreover, retail payments are pre-dominantly one-way from customers to merchants, and thus those links get frequently depleted, reducing the route availability even further~\cite{engelmann2017towards}. 

\mypara{Payment hubs} attempt to solve the route-availability problem by having all
customers establish a payment channel to a central hub that is linked to all merchants. The main problem of this approach is that the hub operator either has to be trusted or it needs to place a very large deposit to guarantee all payments. Since the required collateral is proportional to the number of customers, in our large retailers example, a hub operator will have to deposit \$250M to support $n=1$M customers with an expenditure of $e=\$250$. To cover the cost of locking in such a large mount of funds, the operator is likely to charge substantial payment fees.

\mypara{Commit-chains} aim to improve payment hubs by reducing or eliminating operator collaterals. To do that, they rely on periodic on-chain \emph{checkpoints} that finalize multiple off-chain transactions at once. While this improves throughput, it does not reduce latency, as users still have to wait for the checkpoint to reach finality on-chain~\cite{khalil2018nc,gudgeonsok}. Other commit-chain variants enable instant payments, but require equally large collaterals as payment hubs. Additionally, in such variants users need to monitor the checkpoints (hourly or daily) to ensure that their balance is represented  accurately~\cite{khalilnocust,gudgeonsok}. We focus on retail setting where customers do not have to be constantly or periodically online (recall Section~\ref{subsec:assumptions}), and thus cannot be expected to perform such monitoring.

\mypara{Side-chains} rely on a small number of collectively trusted validators to track and agree on pending transactions. Typically, consensus is achieved using Byzantine-fault tolerant protocols~\cite{castro2002practical} that scale up to a few tens of validators and require 2/3 of honest validators. Thus, side-chains contradict one of the main benefits of permissionless blockchains, the fact that no trusted entities are required. Additionally, BFT consensus requires several communication rounds and has high message complexity.

\subsection{Requirements}
\label{subsec:requirements}

Given these limitations of previous solution, we define the following requirements for our work.

\begin{itemize}
	\item \textit{R1: Fast payments without trusted validators.} Our solution should enable payment recipients such as merchants to accept fast payments assuming no trusted validators. 

	\item \textit{R2: Practical collaterals for large deployments.} Collaterals should be pratical, even when the system scales for large number of customers or many merchants. In particular, the customer collaterals should only depend on their own spending, not the number of merchants. The entities that operate the system (in our solution, the merchants) should deposit an amount that is proportional to their sales, not the number of customers.
		
	\item \textit{R3: Cheap payment processing.} When deployed on top of an existing blockchain system such as Ethereum, payment processing should be inexpensive.
\end{itemize}

\section{\sysname Overview}
\label{sec:overview}

In our solution, \sysname, customer collaterals are held by a smart contract called \emph{Arbiter} and enable merchants to safely accept payments before a transaction reaches finality on the chain. If a pending transaction does not get confirmed on the chain (double spending attack), the victim merchant can recover the lost funds from the customer's collateral. 

For any such solution, it is crucial that the total value of a customer's pending transactions never exceeds the value of its collateral. This invariant is easy to ensure in a single-merchant setup by keeping track of the customers' pending transactions. However, in a retail setting with many merchants, each individual merchant does not have a complete view of each customer's pending payments, as each customer can perform purchases with several different merchants simultaneously. 

A natural approach to address this problem is to assume that the merchants \emph{collaborate} and collectively track pending transactions from all customers (see Figure~\ref{fig:udp}). Below, we outline simple approaches to realize such collaboration and point out their drawbacks that motivate our solution.

\subsection{Strawman Designs}
\label{subsec:strawman}

\mypara{Transaction broadcasting.} Perhaps the simplest approach would be to require that all merchants broadcast all incoming payments, so that everyone can keep track of the pending transactions from each customer. Such a solution would prevent customers from exceeding their collateral's value, but assumes that all the merchants are \emph{honest}. A malicious merchant, colluding with a customer, could mislead others by simply not reporting some of the customer's pending payments. The customer can then double spend on all of its pending transactions, thus enabling the colluding merchant to deplete its collateral and prevent other merchants from recovering their losses. The same problem would arise also in cases where a benign merchant fails to communicate with some of the merchants (e.g., due to a temporary network issue) or if the adversary drops packets sent between merchants.

We argue that such transaction broadcasting might be applicable in specific settings (e.g., large retailers that trust each other and are connected via high-availability links), but it fails to protect mutually-distrusting merchants such as small shops or restaurants, and is unsuited to scenarios where mutually-trusting merchants cannot establish expensive links.

\mypara{Unanimous approval.} Alternatively, merchants could send each incoming transaction to all other merchants and wait until each of them responds with a signed approval.
While some of the merchants may act maliciously and equivocate, rational merchants holding pending transactions from the same customer will not, as this would put their own payments at risk. The Arbiter contract would refuse to process any claims for losses, unless the merchant can provide approval statements from all other merchants.
Such unanimous approval prevents the previous attack (assuming rational actors), but it suffers from poor liveness. Even if just one of the merchants is unreachable, the system cannot process payments.

\mypara{BFT consensus.} A common way to address such availability concerns is to use Byzantine-fault tolerant protocols. For example, the merchants could use a BFT consensus such as~\cite{castro2002practical} to stay up to date with all the pending payments. Such a solution can tolerate up to $1/3$ faulty (e.g., non-responsive) merchants and thus provides increased availability compared to the previous design.
However, this solution has all the limitations of side-chains (recall Section~\ref{subsec:limitations}). In particular, BFT consensus requires $2/3$ trusted validators and several rounds of communication. Therefore, BFT side-chains are not ideal even in the case of mutually-trusting merchants. 

\vspace{-0.1cm}\subsection{\sysname Main Concepts}
\label{subsec:design}

To overcome the problems of the above strawman designs, we use two primary techniques: (a) \emph{majority approval} and (b) \emph{merchant collaterals}, such that fast payments can be accepted safely even if all the merchants are untrusted and some of them are unreachable.

The use of majority approval generates non-deniable evidence about potentially misbehaving merchants. Together with merchant collaterals, this enables attack victims to recover their losses, in case a merchant colludes with a malicious customer. In an example attack, a malicious customer sends transaction $\tx{}$ to a victim merchant and $\tx{}'$ to a colluding merchant, and simultaneously double spends both $\tx{}$ and $\tx{}'$. The colluding merchant claims the lost value from the customer's collateral which exhausts it and prevents the victim merchant from recovering losses from the customer's collateral. However, the victim can use the colluding merchant's signature from the majority approval process as evidence and reclaim the lost funds from the merchant's collateral. Given these main ideas, next we provide a brief overview of how \sysname works. 

\begin{figure}
	\begin{subfigure}[t]{0.50\linewidth}\centering 
	\includegraphics[width=0.93\textwidth]{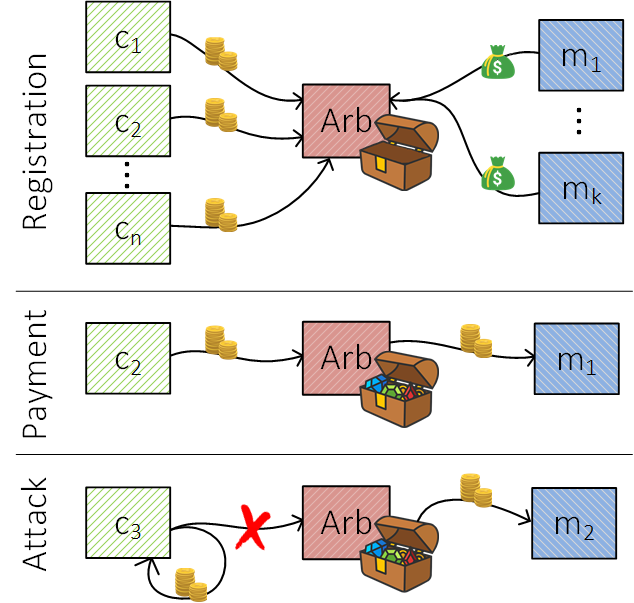}
	\hfill
	\end{subfigure}%
	\begin{subfigure}[t]{0.50\linewidth}\centering
	\includegraphics[width=0.91\textwidth]{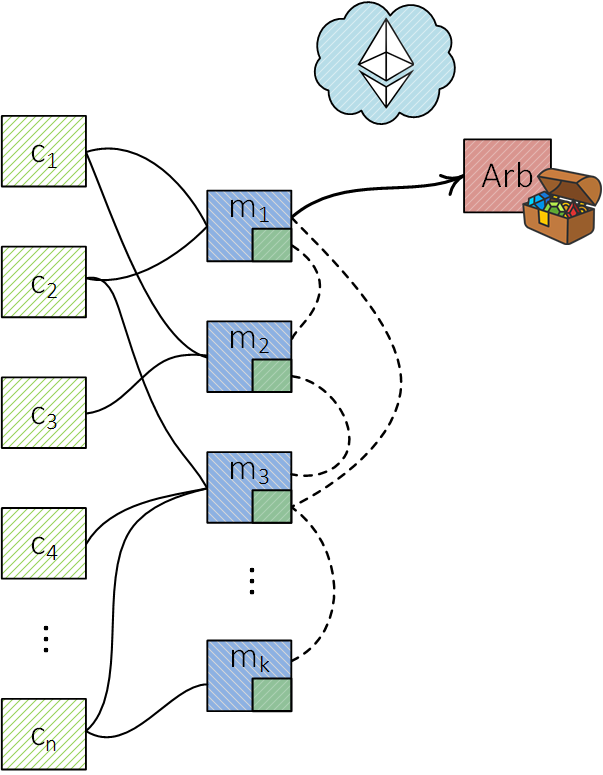}%
	\end{subfigure}
	{\renewcommand{\thefigure}{\arabic{figure} (left)}
	\caption{\textbf{Flow of funds.} Customers and merchants deposit  collaterals to the arbiter smart contract. Payments flow from customers to merchants through the arbiter. In case of attack, the victim merchant can recover any losses from the arbiter.}
	\label{fig:operations}}
	{\renewcommand{\thefigure}{\arabic{figure} (right)}
	\caption{\textbf{Example deployment} where each merchant operates one statekeeper. Customers make payments towards merchants, who consult a majority of the statekeepers before accepting them.}
	\label{fig:udp}}
	\figsaver
\end{figure}

\mypara{Collaterals.} To register in the \sysname system, each customer deposits a collateral to Arbiter's account (see ``Registration'' in Figure~\ref{fig:operations}). The value of the collateral is determined by the customer, and it should suffice to cover its expenditure $e_t$ during the blockchain latency period (e.g., $e_t=\$100$ for 3 minutes in Ethereum). 

Since merchants are untrusted, they also need to deposit a collateral. The size of the merchant collateral depends on the total value of sales that \emph{all} participating merchants process within the latency period. For example, for a consortium of $k=100$ small shops that process $p_t=6$ payments of $e_t=\$5$ on average during the latency period, a collateral of \$3,000 will suffice. In a deployment with $k=100$ larger retailers, where each one handles $p_t=15$ purchases of $e_t=\$100$ (on average) within the latency period, each merchant will need to deposit \$150,000~\cite{tesco2016,tesco2018}. We acknowledge that this is a significant deposit, but feasible for large corporations.

\mypara{Payment approval.} After registering, a customer can initiate a payment by sending a payment intent to a merchant together with a list of its previous approved but not yet finalized transactions. The merchant can verify that the provided list is complete, by examining index values that are part of each \sysname transaction. If the total value of intended payment and the previously approved payments does not exceed the value of the customer's collateral, the merchant proceeds. 

To protect itself against conflicting \sysname transactions sent to other merchants simultaneous, the merchant collects approval signatures from more than half of the participating merchants, as shown in Figure~\ref{fig:udp}. Such majority approval does not prevent malicious merchants from falsely accepting payments, but provides undeniable evidence of such misbehavior that the merchant can later use to recover losses in case of an attack. In essence, a merchant $m$ signing a transaction $\tx{}$ attests that ``$m$ has not approved other transactions from the same customer that would conflict with $\tx{}$''. Thus, $m$ needs to check each transaction against those it has signed in the past. The merchant who is about to accept a payment also verifies that the approving merchants have each sufficient collateral left in the system to cover the approved payment.
Once these checks are complete, the merchant can safely accept the payment. 

The customer construct a complete \sysname payment such that the payment funds are initially sent to Arbiter that logs the approval details into its state and after that forwards the payment value to the receiving merchant (see ``Payment'' in Figure~\ref{fig:operations}). We route all payments through the Arbiter contract to enable the Arbiter to perform claim settlement in blockchain systems like Ethereum where smart contracts do not have perfect visibility to all past transactions (recall Section~\ref{subsec:assumptions}). If future blockchain systems offer better transaction visibility to past transactions for contracts, payments can be also routed directly from the customers to the merchants. 

\mypara{Settlements.} If the transaction does not appear in the blockchain after a reasonable delay (i.e., double-spending attack took place), the victim merchant can initiate a settlement process with the arbiter to recover the lost funds (``Attack in Figure~\ref{fig:operations}). The Arbiter uses the logged approval evidence from its state to determine who is responsible for the attack, and it returns the lost funds to the victim merchant either from the collateral of the customer or the misbehaving merchant who issued false approval.

\mypara{Separation of statekeeping.} So far, we have assumed that payment approval requires signatures from merchants who track pending transactions. 
While we anticipate this to be the most common deployment option, for generality and separation of duties, we decouple the tracking and approval task of merchants into a separate role that we call \emph{statekeeper}. 
For the rest of the paper, we assume a one-to-one mapping between merchants and statekeepers, as shown in Figure~\ref{fig:udp} and in Appendix~\ref{app:singlesk} we discuss alternative deployment options. When there are an identical number of merchants and statekeepers, the value of statekeepers' collaterals is determined as already explained for merchant collaterals.

\mypara{Incentives.} There are multiple reasons why merchants would want to act as statekeepers.  One example reason is that it allows them to join a \sysname consortium, accept fast blockchain payments securely and save on card payment fees. To put the potential saving into perspective, considering our large retail stores example and the common 1.5\% card payment fee (and a cost of \$0.16 per \sysname payment). In such a case, a collateral of \$150,000 is amortized in ${\sim}37$ days. 
Potential free-riding by a consortium member could be handled by maintaining a ratio of approvals made and received for each merchant and excluding merchants who fall below a  threshold.

Another example reason is that the merchants could establish a fee-based system where each approver receives a small percentage of the transaction's value to cover the operational cost and the collateral investment needed to act as a statekeeper.

\section{\sysname Details}
\label{sec:protocols}

\noindent In this section, we describe the \sysname system in detail. We instantiate it for Ethereum, but emphasize that our solution is not limited to this particular blockchain. We start with background on aggregate signatures, followed by \sysname data structures, registration, payment and settlement protocols.

\subsection{Background on BLS Signatures}
\label{subsec:bls}

\noindent To enable signature \emph{aggregation} for reduced transaction size and efficient transaction verification, we utilize the Boneh-Lynn-Shacham (BLS) Signature Scheme~\cite{DBLP:journals/joc/BonehLS04}, along with extensions from~\cite{DBLP:conf/eurocrypt/BonehGLS03,DBLP:conf/eurocrypt/RistenpartY07}. BLS signatures are built on top of a bilinear group and a cryptographic hash function $H: \{0,1\}^{*} \mapsto \Gr$. 
A bilinear group $\bp = (p,\Gr,\Hr,\Tr,e,g,h)$
\footnote{Boneh et al.'s scheme utilises specific bilinear pairings where there is an efficiently computable isomorphism between $\Gr$ and $\Hr$ which are not implemented over Ethereum.  However, Boneh et al.~\cite{DBLP:journals/iacr/BonehDN18} observed that the proof of security does still apply with respect to the more commonly used pairings under a stronger cryptographic assumption. We assume the signature scheme is implemented with regards to bilinear groups with no known isomorphism.} 
consists of cyclic groups $\Gr$, $\Hr$, $\Tr$ of prime order $p$, generators $g$, $h$ that generate $\Gr$ and $\Hr$ respectively, and a \emph{bilinear map} $e:\Gr\times \Hr\rightarrow \Tr$ such that:
The map is bilinear, i.e., for all $u\in \Gr$ and $v\in \Hr$ and for all $a,b\in \Z_p$ we have $e(u^a,v^b)=e(u,v)^{ab}$; The map is non-degenerate, i.e., if $e(u,v)=1$ then $u=1$ or $v=1$.
There are efficient algorithms for computing group operations, evaluating the bilinear map, deciding membership of the groups, and sampling generators of the groups. 

The main operations are defined as follows:

\begin{itemize}
	\item \textit{Key generation:} A private key $x_j$ sampled from $\Z_p$, and the public key is set as $v_j = h^{x_j}$. A zero-knowledge proof of knowledge $\pi_j$ for $x_j$ can generated using a Fischlin transformed sigma protocol~\cite{DBLP:conf/crypto/Fischlin05} or by signing the public key with its private counterpart.
	\item \textit{Signing:} Given a message $m$, the prover computes a signature as $\sigma_j = H(m)^{x_j}$.
	\item \textit{Aggregation:} Given message $m$ and signatures  ${\sigma_1, \ldots, \sigma_n}$, an aggregated signature $A$ is computed as $A = \prod_{j=1}^n \sigma_j$.
	\item \textit{Verification:} The verifier considers an aggregated signature to be valid if $e\left( A,h \right) = e\left(H(m), \prod_{j =1}^{n} v_j\right).$
\end{itemize}

\subsection{Data Structures}

\begin{table}[t]
	\centering
	\resizebox{0.9\columnwidth}{!}{
	\begin{tabular}{l c c l}
	\toprule
	\thead{Field} & \thead{Symbol} & \thead{Type} & 	\thead{Description} \\ \hline 
	\addlinespace
	\textit{To} & $\tx{\textrm{to}}$ & 160-bit & Arbiters addr \\ 
	\textit{From} & $\tx{f}$ & 160-bit & Customer's addr \\ 
	\textit{Value} & $\tx{v}$ & Integer & Transferred funds \\
	\textit{ECDSA Sig.} & $v,r,s$ & 256-bits &  Tx signature triplet. \\ \addlinespace
	\textit{Data} &  &   &  \\ 
	\textit{\ \ $\hookrightarrow$ Operation} & $\tx{\textrm{op}}$ & String & e.g., ``Pay'', ``Claim'' \\ 
	\textit{\ \ $\hookrightarrow$ Merchant} & $\tx{m}$ & 160-bit & Merchant's address \\
	\textit{\ \ $\hookrightarrow$ Payment Index} & $\tx{i}$ & Integer & Monotonic counter\\
	\textit{\ \ $\hookrightarrow$ Signatures} & $\tx{A}$ & 512-bits & Aggregate signature\\ 
	\textit{\ \ $\hookrightarrow$ Quorum} & $\tx{q}$ & 256-bits & Approving parties\\ 
	\bottomrule 
	\addlinespace
	\end{tabular}
	}
	\caption{\textbf{\sysname transaction $\tau$ format.} All \sysname-specific information is encoded into the Data field of Ethereum transactions.}
	\label{tab:txstructure}
	\floattabsaver
\end{table}

\mypara{Transactions.} All interaction with the Arbiter smart contract are conducted through Ethereum transactions. We encode the \sysname-specific information in the to \emph{Data} field, as shown in Table~\ref{tab:txstructure}. Each \sysname transaction $\tau$ carries an identifier of the operation type $\tx{\textrm{op}}$ (e.g., ``Payment'', ``Registration'', ``Claim''). Moreover, transactions marked as ``Payments'' also include the address of the transacting merchant $\tx{m}$, a monotonically increasing payment index $\tx{i}$,\footnote{We note that the \sysname transaction index $\tx{i}$ is a separate field from the standard Ethereum transaction \emph{nonce}. Although both are monotonically increasing counters, the \sysname index counts payments \emph{inside} the \sysname system, while the Ethereum nonce counts all transactions by the same user, also \emph{outside} the \sysname system.}
an aggregate $\tx{A}$ of the statekeepers' approval signatures, and a vector of bits $\tx{q}$ indicating which statekeepers contributed to the aggregate signature. $\tx{v}$ denotes the amount of funds in $\tx{}$.

\begin{table}
	\centering	
	\resizebox{0.9\columnwidth}{!}{
	\begin{tabular}{l l l}
	\toprule
	\thead{Field} & \thead{Symbol} & \thead{Description} \\ \hline 
	\addlinespace

	\textit{Customers} & $C$ & Customers \\
	\textit{\ \ $\hookrightarrow$ entry} & $C[c]$ &  Customer $c$ entry\\
	\textit{\ \ \ \ $\hookrightarrow$ Collateral} & $C[c].\col{c}$ & Customer's collateral \\
	\textit{\ \ \ \ $\hookrightarrow$ Clearance}  & $C[c].cl$ &  Clearance index \\
	\textit{\ \ \ \ $\hookrightarrow$ Finalized}  & $C[c].D$ & Finalized Transactions \\
	\textit{\ \ \ \ \ \ $\hookrightarrow$ entry}  & $C[c].D[i]$ & Entry for index $i$ \\
	\textit{\ \ \ \ \ \ \ \ $\hookrightarrow$ Hash} 	  & $C[c].D[i].H(\tx{})$ & Tx Hash\\
	\textit{\ \ \ \ \ \ \ \ $\hookrightarrow$ Signatures} & $C[c].D[i].\tx{A}$ & Aggregate signature \\
	\textit{\ \ \ \ \ \ \ \ $\hookrightarrow$ Quorum}     & $C[c].D[i].\tx{q}$ & Approving parties \\
	\textit{\ \ \ \ \ \ \ \ $\hookrightarrow$ Bit}        & $C[c].D[i].b$ & Sig. verified flag \\
	\addlinespace

	\textit{Merchants} & $M$ & Merchants \\
	\textit{\ \ $\hookrightarrow$ entry} & $M[m]$ & Merchant $m$ entry \\
	\addlinespace

	\textit{Statekeepers} & $S$ &  Statekeepers \\
	\textit{\ \ $\hookrightarrow$ entry} & $S[s]$ & Statekeeper $s$ entry\\
	\textit{\ \ \ \ $\hookrightarrow$ Allocation} & $S[s].\col{s}[m]$ & Value per merchant \\

	\bottomrule 
	\addlinespace
	\end{tabular}
	}
	\caption{\textbf{Arbiter's state.} The arbiter smart contract maintains a record for each customer, merchant and statekeeper.}
	\label{tab:arb-state}
	\floattabsaver
	
\end{table}

\mypara{Arbiter state.} The Arbiter maintains state, shown in Table~\ref{tab:arb-state}, that consist of the following items: a \emph{Customers} dictionary $C$ that maps customer public keys to their deposited collateral $\col{c}$, a clearance index $cl$, and a dictionary $D$ of finalized transactions. $D$ maps transaction index values to tuples that contain the hash of the transaction H($\tx{}$), $\tx{A}$, $\tx{q}$, and a bit $b$ indicating whether the signature has been verified. 

\mypara{Customer state.} Each customer $c$ maintains an ordered list $\Listone$ with its past \emph{approved} transactions.  We denote $\mathcal{S}_c$ the subset of those transactions from $\Listone$ that are still pending inclusion in the blockchain ($\state{c} \subseteq \Listone$).

\mypara{Statekeeper state.} Each statekeeper maintains a list $\Listtwo$ of all the payment intents it has approved for each customer $c$.
 
\mypara{Merchant state.} Each merchant maintains a table $R[s]$ indicating how much collateral of each statekeeper $s$ it can claim. Initially, $R[s]=\col{s}/k$ for each statekeeper.

\subsection{Registration}
\label{sec:registration}

\noindent Customers can join and leave the system at any time, but the set of merchants and statekeepers is fixed after the initialization (in Appendix~\ref{app:generalizations} we discuss dynamic sets of merchants and statekeepers).

\mypara{Merchant registration.} Merchant $m$ submits a registration request which is a transaction that calls a function of the Arbiter contract. Once executed, the Arbiter adds a new entry $M[m]$ to its state, where $m = \tx{f}$ is the merchant's account address (public key).

\mypara{Customer registration.} Customer $c$ sends a transaction to the Arbiter with the funds
that it intents to deposit for its collateral $\col{c}$. The Arbiter contract creates a new customer entry $C[c]$, sets the customer's collateral to $\tx{v}$, and initializes the clearance index to $cl=0$. The Arbiter also initializes a new dictionary $C[c].D$ to log the registered customer's payments.

\mypara{Statekeeper registration.} The registration process for statekeepers is similar to that of customers. However, a statekeeper's registration includes also a proof of knowledge for their BLS private key $x_j$ to prevent \emph{rogue key attacks} (see~\cite{ristenpart2007power} for details). The proof is included in the data field of the registration transaction and is verified by the Arbiter. 
Additionally, the statekeeper may also define how its collateral is to be distributed between the merchants. The maximum value each merchant can claim is $S[s].d$. We consider the case where each statekeeper collateral $\col{s}$ is equally allocated between all $k$ merchants, and thus each merchant can recoup losses up to $\col{s}/k$ from statekeeper $s$.

\subsection{Payment Approval}
\label{subsec:payments}

\noindent The payment-approval protocol proceeds as follows (Figure~\ref{fig:fastpayment}):

\begin{figure}
	\centering
	\includegraphics[width=0.9\linewidth]{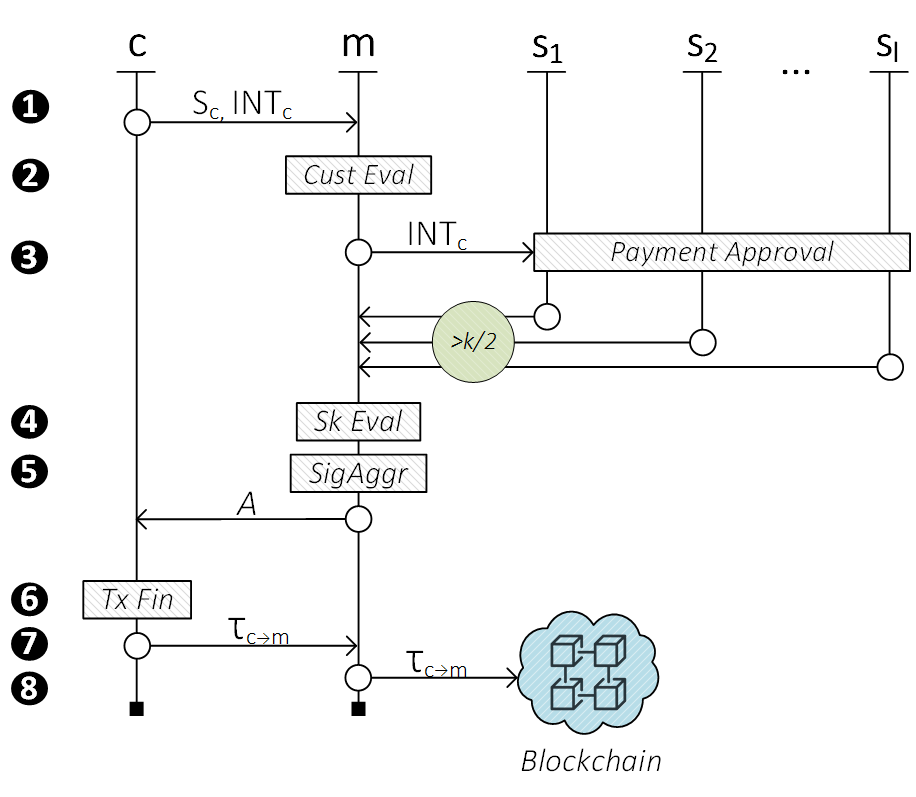}
	\caption{\textbf{Payment-approval protocol.} Customer $c$ initiates a payment to merchant $m$, who requests a majority approval from statekeepers $s_1, ..., s_l$. Merchant $m$ aggregates the received responses and forwards them back to customer $c$, who creates and broadcasts the final transaction $\tx{}$ transferring the funds to the merchant.}
	\label{fig:fastpayment}
	\figsaver
\end{figure}

\smallskip
\noindent\emph{\circled{1} Payment initialization.} To initiate a payment, customer $c$ creates a payment intent INT$_c$ and sends it to merchant $m$, together with its state $\mathcal{S}_c$ that contains a list of already approved but still pending payments. The intent carries all the details of the payment (e.g., sender, recipient, amount, index) and has the same structure as normal transaction, but without the signatures that would make it a complete valid transaction.

\smallskip
\noindent\emph{\circled{2} Customer evaluation.} Upon reception, merchant $m$ checks that all the index values in the interval $\{1 \ldots \textrm{INT}_c[i]\}$ appear either in the blockchain as finalized transactions or in $\mathcal{S}_c$ as approved but pending transactions. If this check is successful, $m$ proceeds to verify that $\col{c}$ suffices to compensate all of the currently pending transactions of that customer (i.e., $\col{c} \geq \sum{\mathcal{S}_c} + \textrm{INT}_c[v]$). Finally, merchant $m$ verifies the approval signatures of the customer's transactions both in $\state{c}$ and in the blockchain.

\smallskip 
\noindent\emph{\circled{3} Payment approval.} Next, the payment intent must be approved by a majority of the statekeepers. The merchant forwards INT$_c$ to all statekeepers, who individually compare it with $c$'s payment intents they have approved in the past (list $\Listtwo$ in their state). If no intent with the same \emph{index} $\tau_i$ is found, statekeeper $s$ appends INT$_c$ in $\Listtwo$, computes a BLS signature $\sigma_s$ over INT$_c$, and sends $\sigma_s$ back to the merchant. If the statekeeper finds a past approved intent from $c$ with the same index value, it aborts and notifies the merchant. 

\smallskip 
\noindent\textit{\circled{4} Statekeeper evaluation.} Upon receiving the approval signatures from a majority of the statekeepers, merchant $m$ uses $R[s]$ to check that each approving statekeeper $s$ has sufficient collateral remaining to cover the current payment value $\tx{v}$. In particular, the merchant verifies that none of the approving statekeepers have previously approved pending payments, whose sum together with the current payment exceeds $R[s]$. (Merchants can check the blockchain periodically and update their $R$ depending on the pending payments that got finalized on-chain.) 
Recall that initially $R[s]$ is set to $\col{s}/k$ for each statekeeper, but in case $m$ has filed a settlement claims in the past against $s$, the remaining collateral allocated for $m$ may have been reduced. This check ensures that in case one or more statekeepers equivocate, $m$ is able to recover its losses in full.

\smallskip
\noindent\textit{\circled{5} Signature aggregation.} 
If the statekeepers evaluation succeeds, $m$ aggregates the approval signatures ${\sigma_1, \ldots, \sigma_{\ceil{(k+1)/2}}}$ (i.e, $A = \prod_{j=1}^{\ceil{(k+1)/2}} \sigma_j$) and sends to customer $c$ the resulting aggregate $A$ and a bit vector $q$ indicating which statekeepers' signatures were included in $A$. Otherwise, if one or more statekeepers do not have sufficient collateral
to approve INT$_c$, the merchant can either contact additional statekeepers or abort.

\smallskip
\noindent\emph{\circled{6} Transaction finalization.} The customer signs and returns to $m$ a transaction $\tx{}$ with the details of INT$_c$, as well as $\tx{to}=$ Arbiter, $\tx{m}=m$, $\tx{A}=A$ and $\tx{q}=q$. 

\smallskip
\noindent\emph{\circled{7} Payment acceptance.} Merchant $m$ first verifies the details and the signatures of the transaction $\tx{}$, and then hands over the purchased goods or service to the customer. Finally, $m$ broadcasts $\tx{}$ to the blockchain network.
(Merchants have no incentive to withhold payments, but if this happens, the customer can simply broadcast it after a timeout.)

\smallskip
\noindent\emph{\circled{8} Payment logging.} Once $\tx{}$ is added in a block, the Arbiter records the payment
and forwards its value to merchant $m$ (Algorithm~\ref{alg:paymentforward}).
To keep the payments inexpensive in Ethereum, our smart contract does \emph{not} verify the aggregate signature $\tx{A}$ during logging. Instead, it performs the expensive signature verification only in cases of disputes, which are expected to be rare. In Section~\ref{sec:analysis}, we show that such optimization is safe.

\begin{algorithm}[t]
	\footnotesize
	\SetAlgoLined
	\executed{Arbiter (smart contract)}
	\SetKwInOut{Input}{Input}
	\SetKwInOut{Output}{Output}
	\Input{\sysname transaction $\tx{}$}
	\Output{None}
	$c \gets \tx{f}$\\
	\If{$c \in C \textrm{ and } \tx{i} \not\in C[c].D$}{
		$h \gets H(\tx{f},\tx{to},\tx{v},\tx{i})$\\
		$C[c].D[\tx{i}] \gets \langle h,\tx{A},\tx{q}, 0\rangle$\\
		$\textrm{Send}(m,\tx{v})$ \Comment{Forward the tx's funds to merchant}\\		

	}\Else{
		$\textrm{Send}(s,\tx{v})$	 \Comment{Return funds to sender}\\
	}
 \caption{\textbf{Record-and-Forward.} Arbiter records the transaction and forwards the payment value to the merchant.}
 \label{alg:paymentforward}
\end{algorithm}

\subsection{Claim Settlement}
\label{sec:settlement}

If a transaction does not appear in the blockchain within a reasonable time, the affected merchant may inspect the blockchain to determine the cause of the delay. There are three possible reasons for a transaction
not getting included in the blockchain in a timely manner:

\begin{enumerate}
	\item \textit{Benign congestion.} A transaction may be of lower priority for the miners compared to other transactions awarding higher miner fees.
	\item \textit{Conflicting transaction.} Another transaction by the same customer prevents the transaction from being included in the blockchain (doubles-pending attack). In Ethereum, two transactions from the same customer conflict when they share the same \textit{nonce} value~\cite{wood2014ethereum}. 
	\item \textit{Customer's account depletion.} A recent transaction left insufficient funds in the customer's account. Ethereum processes the transactions of each customer in increasing nonce order, and thus transactions with greater \textit{nonce} values may be invalidated by those preceding them.
\end{enumerate}

\smallskip
In the first case, the merchant must simply wait and check again later. In the latter two cases, 
the merchant can initiate a settlement claim to recoup the lost value $\tx{v}$. Such claims are sent to the Arbiter and include: (1) the pending transaction $\tx{}^p$ for which the merchant requests settlement, (2) a list $\mathbb{T}_p$ with all preceding (lower index) and pending \sysname transactions, and (3) optionally a list of conflicting transaction tuples $\mathbb{T}_{\textrm{cnfl}}$. The merchant can construct $\mathbb{T}_p$ from the $\state{c}$ (received at payment initialization) by removing the transactions that are no longer pending ($\mathbb{T}_p \subseteq \state{c}$).

\begin{algorithm}[t]
	\footnotesize
	\SetAlgoLined
	\executed{Arbiter (smart contract)}
	\SetKwInOut{Input}{Input}
	\SetKwInOut{Output}{Output}
	\Input{Pending transaction $\tx{}^{p}$ \\ Ordered list of pending transactions $\mathbb{T}_p$ \\ List of conflicting transaction tuples $\mathbb{T}_{\textrm{cnfl}}$ (optional)}
	\Output{None}
	\If{Verify($\tx{A}^p$)}{
	    $r \gets \textrm{Claim-Customer(}\tx{}^{p},\mathbb{T}_p$)\\
		\If {$r > 0$}{
			$\textrm{Claim-Statekeeper(}\tx{}^p,\mathbb{T}_{\textrm{cnfl}}, r$)\\
		}
		$c \gets \tx{f}^p$\\
		$C[c].D[\tx{i}] \gets \langle h,\tx{A},\tx{q}, 1\rangle$ \Comment{Log tx as processed}\\
	}
 \caption{\textbf{Claim-Settlement.} Arbiter processes settlement claim using customers' and statekeepers' collaterals.} \label{alg:claim}
\end{algorithm}

\mypara{Settlement overview.} Algorithm~\ref{alg:claim} describes how the Arbiter executes settlement. 
First, it checks the majority approval by verifying the BLS signature aggregate $\tx{A}^p$.
Then, it tries to recoup the transaction value $\tx{v}^p$ from the customer's collateral using the \emph{Claim-Customer sub-routine}. If the collateral does not suffice to cover the transaction value $\tx{v}$, it proceeds to claim the remaining amount from the collateral of the equivocating statekeeper using the \emph{Claim-Statekeeper sub-routine}. Finally, the contract logs $\tx{i}$ so that no further claims from the customer's collateral can be made for the same index $\tx{i}$.

\begin{algorithm}[t]
	\footnotesize
	\SetAlgoLined
	\SetKwInOut{Executed}{Actor}
	\SetKwInOut{Input}{Input}
	\SetKwInOut{Output}{Output}
	\Executed{Arbiter (smart contract)}
	\Input{Pending transaction $\tx{}^{p}$ \\ Preceding pending transactions $\mathbb{T}_p$}
	\Output{Residual $r$ or $\bot$}
	
	 $c \gets \tx{f}^{p}$\\
     $\textrm{I}^* \gets C[c].D$ \Comment{Pass by reference} \\

    \tcc{Verify signatures of preceding non-pending txs.}
	 \For {$\forall \{i \in \textrm{I}^*|\textrm{I}^*[i].b=0\}$}{ \label{line:past_verify1}
		\If {Verify($\textrm{I}^*[i]_A$)}{
			$I[i].b \gets 1$\\
		}\Else{
			$\textrm{del } I^*[\textrm{i}]$  \Comment{Past tx had no approval.}\label{line:past_verify2}\\
		}
	 }

   \tcc{Any pending \& preceding txs missing?}
     \If{$\exists\, i \in \{1 \dots \tx{i-1}^{p} \}$ such that $i \not\in \textrm{I}^* \textrm{ and } i \not\in \mathbb{T}_p$}{ \label{line:tp_complete1}
		\Return $\bot$\\  \label{line:tp_complete2}
	 }

    \tcc{Verify signatures of pending preceding txs.}
  	\If {$!\textrm{Verify(}\mathbb{T}_p \textrm{)}$}{ \label{line:tp_verify1}
			\Return $\bot$\\						  \label{line:tp_verify2}
	  }

    \tcc{Process claim.}
	 \If{$\tx{i}^{p} \not\in \textrm{I}^*$}{\label{line:not_processed2}
		 $\col{c}^* \gets C[c].\col{c}$ \Comment{Pass by reference}\\
         $cov \gets max(0,\col{c}^* - \sum \mathbb{T}_p)$ \Comment{Max claimable amount}\\  \label{line:max_claim}
	 	 $\rho \gets min(cov, \tx{v}^{p})$\\
	     $\col{c} \gets \col{c}^* - \rho$\\
   	     $\textrm{Send(}\tx{m}^{p},\rho$)\\
	     $\textrm{r} \gets \tx{v}^{p}-\rho$\\
	 	\Return{$r$}\\
	 	
	 	\label{line:return-overexp}
     }
     \Else{
	 	\Return{$\tx{v}$} \Comment{Statekeeper equivocated}\\
	 	\label{line:sk-equi}
	 }

 \caption{\textbf{Claim-Customer.} Arbiter recovers lost funds from the customer's collateral or returns the remaining amount.} \label{alg:claim_cust}
\end{algorithm}

\mypara{Claim-Customer} sub-routine is shown in Algorithm~\ref{alg:claim_cust}. First (lines~\ref{line:past_verify1}-\ref{line:past_verify2}), it verifies the signatures of preceding finalized transactions that have not been verified already (recall that signatures are not verified during Record-and-Forward to reduce payment processing cost). 
Subsequently, the contract checks that the list $\mathbb{T}_p$ contains all the pending transactions preceding $\tx{}^p$ (lines~\ref{line:tp_complete1}-\ref{line:tp_complete2}) and verifies their signatures (lines~\ref{line:tp_verify1}-\ref{line:tp_verify2}). For every transaction that
is deleted (line~\ref{line:past_verify2}), the merchant should have included another transaction with the same index and valid approval signature in $\mathbb{T}_p$. These two checks ensure that the merchant did its due diligence and verified the approval signatures of preceding pending and non-pending transactions during "Customer evaluation" step (Section~\ref{subsec:payments}).

After that, the contract checks for another valid transaction with the same index, and if there was it returns $\tx{v}^p$ to its parent process (line \ref{line:sk-equi}). In this case, the losses should be claimed from the equivocating statekeepers' collaterals. Otherwise, the arbiter compensates the merchant from the customer's collateral $\col{c}$. In case the customer's collateral does not suffice to fully recoup $\tx{v}^p$, then the algorithm returns the remaining amount (line \ref{line:return-overexp}). This amount is then used in further claims against the statekeepers. 

The \sysname system supports arbitrarily many past transactions. However, in practice, an expiration window can be used to reduce the computational load of verifying the signatures of preceding transactions. For example, a 24-hour window would allow sufficient time to accept payments and claim settlements if needed, and at the same time keep the number of past-transaction verifications small. Such time window would also reduce the storage requirements of the Arbiter smart contract, as it would only need to store the most recent transactions from each user in $C[c].D$, instead of all the past transactions. Note that valid payments will appear in the blockchain within a few minutes and thus the operators' collateral will be quickly freed to be allocated to other transactions.

\mypara{Claim-Statekeepers} sub-routine is shown in Algorithm~\ref{alg:claim_sk}. It is executed in cases where the customer's collateral does not suffice to fully recoup the value of the transaction $\tx{v}^p$ for which settlement is requested. 
In this case, the arbiter attempts to recover the lost funds from the equivocating statekeepers. 

The Arbiter iterates over the tuples of conflicting (and preceding) transactions until $\tx{v}^p$ has been fully recovered or there are no more tuples to be processed (line~\ref{line:process_tuple}). For each of those tuples, the Arbiter does the following: First, it verifies the approval signatures of the two transactions, checks that the two transactions precede $\tx{}^p$ and that they are not identical (lines~\ref{line:verify_tuple1}-\ref{line:verify_tuple2}). Then, it uses \textit{FindOverlap}() based on bit strings $\tx{q}$ of processed transactions
to identify the equivocating statekeepers (line~\ref{line:find_overlap}). Finally, it computes the amount that the merchant was deceived for ($|\tx{v}'-\tx{v}''|$) and subtracts it from the collateral of one of the statekeepers (lines~\ref{line:compute_amount1}-\ref{line:compute_amount2}). 

We only deduct the missing funds from one of the equivocating statekeepers in each tuple, as our goal is to recoup the $\tx{v}^p$ in full. However, \sysname can be easily modified to punish all the equivocating statekeepers. Recall that each merchant is responsible for ensuring that the collateral allocated by each statekeeper for them, suffices to cover the total value of pending payments approved by that statekeeper (line~\ref{line:enough_col}).

\begin{algorithm}[t]
	\footnotesize
	\SetAlgoLined
	\executed{Arbiter (smart contract)}
	\SetKwInOut{Input}{Input}
	\SetKwInOut{Output}{Output}
	\Input{Residual $r$ \\ Pending Transaction $\tx{}^p$ \\ Conflicting transaction tuples $\mathbb{T}_{\textrm{cnfl}}$}
	\Output{None}
	
	$\textrm{left} \gets r$\\
	\While{$\textrm{left} > 0 \textrm{ and } \mathbb{T}_{\textrm{cnfl}} \neq \emptyset$}{ \label{line:process_tuple}
		
    	$\langle \tx{}', \tx{}''\rangle \gets \mathbb{T}_{\textrm{cnfl}}\textrm{.pop()}$\Comment{Tuple of txs with the same idx}\\

	  	\If{$\textrm{Verify(}\tx{}' \textrm{)} \textrm{ and } \textrm{Verify(}\tx{}''\textrm{) } \textrm{and}$\\ \label{line:verify_tuple1}
			$\ \ \ \ \tx{i}' \leq \tx{i}^{p} \textrm{ and } \tx{i}'' \leq \tx{i}^{p} \textrm{ and } \tx{}' \neq \tx{}''$}{ \label{line:verify_tuple2}
			$sk \gets \textrm{FindOverlap(}\tx{}',\tx{}''\textrm{)}$\Comment{Find who equivocated}\\ \label{line:find_overlap}
			\If{$sk \neq\emptyset$}{ 
				$\delta \gets |\tx{v}'-\tx{v}''|$\\ \label{line:compute_amount1}
				 $\col{sk}^* \gets S[sk].col_sk[\tau^p_m]$ \Comment{Pass by reference}\\
				\tcc{Is there enough collateral left?}
				\If{$\col{sk}^*[m] - \delta \geq 0$}{\label{line:enough_col}
					$\col{sk}^*[m] \gets \col{sk}^*[m] - \delta$ \\ \label{line:compute_amount2}
        	    }
   	    	    $\textrm{left} \gets \textrm{left} - \delta$\\
			}
		}
	}
	
	$\rho \gets min(\tx{v}^{p},|\tx{v}^{p}-\textrm{left}|)$\\
	$\textrm{Send(}\tau^p_m,\rho$)\\

 \caption{\textbf{Claim-Statekeeper.} Arbiter sends lost funds from the misbehaving statekeepers' collaterals to the affected merchant.} 
 \label{alg:claim_sk}
\end{algorithm}

\subsection{De-registration}

Customers can leave the \sysname system at any point in time, but in the typical usage de-registrations are rare operations (comparable to credit card cancellations). The process is carried out in two steps, first by clearing the customer's pending transactions and subsequently by returning any remaining collateral. This two-step process allows enough time for any pending settlements to be processed before the collateral is returned to the customer.

The customer submits a clearance request that updates the clearance field $C[c].cl$ to the current block number.
At this point, the customer's account enters a clearance state that lasts for a predetermined period of time (e.g., 24 hours). During this time, customer $c$ can no longer initiate purchases, but $\col{c}$ can still be claimed as part of a settlement. Once the clearance is complete, the customer can withdraw any remaining collateral by submitting a withdrawal request. The Arbiter checks that enough time between the two procedures has elapsed and then returns the remaining $\col{c}$ to customer $c$.

\section{\sysname Analysis}
\label{sec:analysis}
In this section, we analyze the safety and liveness properties of \sysname.

\subsection{Safety}
First we prove that a merchant who follows the \sysname protocol to accept a fast payment is guaranteed to receive the full value of the payment (our Requirement R1) \update{1}{given the strong adversary defined in  Section~\ref{subsec:assumptions}.}

\mypara{Definitions.} We use $\mathcal{BC}$ to denote all the transactions in the chain, and $\mathcal{BC}[c]$ to refer to the subset where $c$ is the sender. We say that a customer's state $\mathcal{S}_c$ is \emph{compatible} with a transaction $\tx{}$ when all $  \tau'$ such that $ \tau' \notin \mathcal{BC}$ and $ \tau'_i < \tau_i$,  $\tau \in \mathcal{S}_c$ and $\sum \mathcal{S}_c + \tx{v}\leq \col{c}$. Less formally, compatibility means that $\mathcal{S}_c$ should include all of $c$'s pending \sysname transactions that precede $\tx{}$, 
while their total value should not exceed $\col{c}$.

\begin{theorem}\label{trm:maintheorem}
Given a customer state $\mathcal{S}_c$ that is \emph{compatible} with a transaction $\tx{}$ approved by a majority of the statekeepers with sufficient collaterals, merchant $\tx{m}$ is guaranteed to receive the full value of $\tx{}$. 
\end{theorem}

\begin{proof}
A \sysname transaction $\tx{}$ transfers funds from customer $c$ to merchant $m$. Here, we use $\col{c}$ to denote the value of the $c$'s collateral when $\tx{}$ was approved, and $\colprime{c}$ to denote the collateral's value at a later point in time. Note that $\colprime{c}$ is potentially smaller than $\col{c}$ as several settlement claims may have been processed in the meantime.

To prove Theorem~\ref{trm:maintheorem}, we consider both the case where $\tx{}$ is included in the blockchain $\mathcal{BC}$, and the case it is not. 
%
In the former case ($\tx{} \in \mathcal{BC}$), merchant $m$ is guaranteed to receive the full payment value $\tx{v}$, as the adversary cannot violate the integrity of the Arbiter smart contract or prevent its execution. Once transaction $\tx{}$ is in the blockchain, the Arbiter contract executes Record-and-Forward  (Algorithm~\ref{alg:paymentforward}) that sends the full payment value $\tx{v}$ to merchant $m$.

In the latter case ($\tx{} \not\in \mathcal{BC}$), the merchant sends a settlement claim request that causes the Arbiter contract to execute Claim-Settlement (Algorithm~\ref{alg:claim}). There are now two cases to consider: either $\colprime{c}$ suffices to cover $\tx{v}$ or it does not.
In the first case ($\tx{v} \leq \colprime{c}$), merchant $m$ can recover all lost funds from 
$\colprime{c}$. For this process, $m$ needs to simply provide the approved $\tx{}$ and the list of \sysname transactions from customer $c$ that preceded $\tx{}$ and are still pending. 
The Arbiter contract verifies the validity of the inputs and sends the payment value $\tx{v}$ to merchant $m$ from $\colprime{c}$.
In the latter case ($\tx{v} > \colprime{c}$), merchant $m$ can still recoup any lost funds, as long as $\mathcal{S}_c$ was \emph{compatible} with $\tx{}$ and $m$ followed the \sysname payment approval protocol and verified before payment acceptance that the approving statekeepers' remaining collaterals $R[s]$ suffice to cover $\tx{v}$ and all other pending transactions previously approved by the same statekeeper.

According to Lemma~\ref{lmm:coldepl_doubleindex}, if $\tx{v} > \colprime{c}$, then there are more than one approved transactions from customer $c$ with the same \textit{index} value $\tx{i}$. As shown by Proposition~\ref{lmm:doubleindex_sk_eq}, for this to happen, one or more statekeepers need to approve two transactions with the same index (i.e., equivocate).
The arbiter contract can find the equivocating statekeepers by comparing the quorum bit vectors $\tx{q}$ and $\tx{q}'$ from the conflicting transaction tuples, and recoups all lost funds from their collaterals.\footnote{Lines~\ref{line:tp_complete1}-\ref{line:tp_complete2} in the Claim-Customer sub-routine (Algorithm~\ref{alg:claim_cust}) force any actor who claims $\tx{v}$ to publish a list with transactions that are pending and have \textit{indices} preceding $\tx{i}$. This enables other merchants to identify conflicting transaction pairs, find the equivocating statekeepers and claim from their collaterals.}

In \sysname, it is the responsibility of the merchant to verify that each statekeeper who approves a payment has sufficient collateral remaining. Before payment acceptance, the merchant verifies that the sum of all approved but pending payments and the current payment are fully covered by $R[s]$ which is initially $\col{s}/k$ but can be reduced by previous settlement claims. Since one merchant cannot claim collateral allocated for another merchant, the merchant is guaranteed to be able recover the full payment value, even if one or more statekeepers equivocate to several merchants simultaneously.
\end{proof}

\noindent\emph{Lemmas.} We now provide proofs for Lemma~\ref{lmm:coldepl_doubleindex} and Proposition~\ref{lmm:doubleindex_sk_eq} used above.

\begin{lemma}\label{lmm:coldepl_doubleindex}
Let $\mathcal{S}_c$ be state of customer $c$ that is \emph{compatible} with a transaction $\tx{}$,
if at any point $\colprime{c} < \tx{v}$ and $\tx{}$ has not been processed, then there exists an approved $\tx{}'$ such that $\tx{}'\not\in \mathcal{S}_c$, and 
$\tx{}'$ has the same index either with $\tx{}$ or a transaction in $\mathcal{S}_c$.
\end{lemma}

\begin{proof}
Let $\colprime{c}$ be the funds left in the customer's collateral after a set of claims $\Pi$ were processed i.e., $\colprime{c} = \col{c} - \sum \Pi$. Let's now assume that $\colprime{c}<\tx{v}$, but there is no $\tx{}'$ that (1) is not in  $\mathcal{S}_c$ and (2) has the same index with $\tx{}$ or with a transaction in $\mathcal{S}_c$. In other words, let's assume that no two approved transactions have the same index $i$.

From $\tx{v} > \colprime{c}$, it follows that $\sum \Pi + \tx{v} > \col{c}$. By definition $\sum \mathcal{S}_c + \tx{v} \leq \col{c}$. Thus, it follows that:
$$\sum \mathcal{S}_c + \tx{v} < \sum \Pi + \tx{v} \Rightarrow \sum \mathcal{S}_c < \sum \Pi$$
From this, we derive that $\Pi \not\subset \mathcal{S}_c$.
Since, $\Pi$ is not a subset of $\mathcal{S}_c$, there is at least one transaction $\tx{}' \in \Pi$ such that $\tx{}' \not\in \mathcal{S}_c$.
$\mathcal{S}_c$ is compatible with $\tx{}$, and thus $\mathcal{S}_c$ contains \textit{all} the pending transactions up to $\tx{i}$. As a result, a pending $\tx{}'$ that is not included in $\mathcal{S}_c$ must have $\tx{i}'$ greater than $\tx{i}$. 
Note that if $\tx{}'$ is not pending, the Arbiter contract will not process the settlement claim (line~\ref{line:not_processed2} in  Algorithm~\ref{alg:claim_cust}). Thus, $\tx{i}<\tx{i}'$.

According to Algorithm~\ref{alg:claim_cust}, any claim for $\tx{}'$ should include all transactions preceding $\tx{}'$ that are still pending. Since, $\tx{}$ is pending and $\tx{i}<\tx{i}'$, the claim should also include $\tx{}$.
However, Line~\ref{line:max_claim} in Algorithm~\ref{alg:claim_cust} ensures that enough funds are left in the customer's collateral to cover all the preceding-pending transactions. This covers $\tx{}$ too, and contradicts our starting premise $\tx{v} > \colprime{c}$.
We have so far shown that:

\renewcommand{\labelenumi}{(\arabic{enumi})}
\begin{enumerate}
	\item If $\colprime{c} < \tx{v}$, then the arbiter processed a claim for a transaction $\tx{}'$ that is not included in $\mathcal{S}_c$.
	\item The collateral will always suffice to cover $\tx{}$, 
even if other claims for transactions with greater index values (i.e., $\tx{i}'>\tx{i}$)
are processed first.
\end{enumerate}

From (1) and (2), it follows that $\tx{i}' \leq \tx{i}$.
\end{proof}

\begin{prop}\label{lmm:doubleindex_sk_eq}
For any two \textit{majority} subsets $M_1 \subseteq S$ and $M_2 \subseteq S$,
where $S$ is the set of all the statekeepers, it holds that $M_1 \cap M_2 \neq \emptyset$.
\end{prop}
\begin{proof}
The proof of this proposition follows directly from the Pidgeonhole principle 
(or Dirichlet's box principle) 
which states that if $n$ items are put in $m$ containers for $n>m$ then at least one container contains multiple items.
\end{proof}

\mypara{No penalization of honest statekeepers.} Above we showed that an honest merchant who accepts a \sysname payment will always receive his funds. Additionally, \sysname ensures that an honest statekeeper cannot be penalized due to benign failures such as network errors or crashes. As outlined in Section~\ref{sec:protocols}, a merchant is penalized only if it approves two conflicting transactions (same index, same customer). This simple policy is sufficient to protect all honest merchants in the case of benign failures.

\mypara{Privacy}
Payment privacy is largely orthogonal to our solution and  inherited from the underlying blockchain. For example, if \sysname is built on Ethereum, for most parts \sysname customers and merchants receive the level of privacy protection that Ethereum transactions provide. 
In Section~\ref{sec:discussion}, we discuss privacy in more detail.

\subsection{Liveness}
\label{sec:liveness}

Next, we explain the liveness property of \sysname.
Payment processing is guaranteed when at least $\ceil{(k+1)/2}$ of the statekeepers are reachable and responsive. If we assume that all the statekeeping nodes are equally reliable then each one of them should have an availability of only $\ceil{(k+1)/2}$. 
\update{1}{We note that \sysname ensures safety and liveness under slightly different conditions. In that regard Snappy is similar to many other payment systems where liveness requires that the payment processor is reachable but the same is not needed for payment safety.}

\section{\sysname Evaluation}
\label{sec:eval}

In this section we evaluate \sysname and compare it with previous solutions.

\subsection{Latency}
\label{subsec:performance}

\noindent In \sysname, the payment approval latency depends on two factors: (a) the number of approving statekeepers and 
(b) the speed/bandwidth of the network links between the merchants and the statekeepers. The number of customers 
and merchants has no effect on the approval latency and its impact on the collateral size is discussed in Section~\ref{sec:scalability}.

To evaluate our Requirement R1 (fast payments), we simulated a setup with several globally-distributed statekeepers and merchants running on Amazon EC2 instances. Both the statekeepers and the merchants were implemented as multi-threaded socket servers/clients in Python 3.7 and used low-end machines with $2$ vCPUs and $2$ GB of RAM. We distributed our nodes in $10$ geographic regions ($4$ different locations in the US, $3$ cities in the EU, and $3$ cities in the Asia Pacific region).

As seen in Figure~\ref{fig:tx}, we tested the payment approval latency for different numbers of statekeepers and various rates of incoming requests. In our first experiment, we initialized 100 merchants who collectively generate 1,000 approval requests per second. We observe that for smaller consortia of up to 40 statekeepers (i.e., at least 21 approvals needed), \sysname approves the requests within 300ms, while larger consortia require up to 550ms on average. This shows that the approval latency increases sub-linearly to the number of statekeepers and experimentally validates the low communication complexity of \sysname.
In particular, the latency doubles for a 5-fold increase in the consortium's size i.e., 40 statekeepers require ${\sim}$300ms to collectively approve a request while 200 statekeepers require ${\sim}$550ms. We also tested our deployment for higher loads: 2,500 and 5,000 requests per second respectively. Our measurements show a slightly increased latency due to the higher resource utilization at the statekeeping nodes. However, the relationship between the number of statekeepers and the approval latency remains 
sub-linear. \update{5}{We observe, in Figure~\ref{fig:tx}, that the variance increases both with the number of statekeepers and the throughput. However, in all cases the payment approvals were completed in less than a second.}

These results demonstrate the practicality of \sysname as it remains under the 2-second approval latency mark that we consider a reliable user experience benchmark~\cite{visatime1,visatime2,emvtime1}.
The measured latencies are consistent with past studies measuring the round-trip times (RTT) worldwide~\cite{hoiland2016measuring} and within the Bitcoin network~\cite{fadhil2016bitcoin,ben2018vivisecting}. In particular, the majority of the bitcoin nodes have an RTT equal or less to 500ms, while only a tiny fraction of the network exhibit an RTT larger than 1,5 seconds. 

An optimized \sysname deployment that uses more capable machines will likely achieve better timings for the aforementioned loads and even larger statekeeping consortia. We did not perform such performance optimizations, as \sysname is best suited to deployments where the number of statekeepers is moderate (e.g., $k=100$). For larger consortia, the statekeeper collaterals grow large (Section~\ref{subsec:collaterals}) and a centralized-but-trustless deployment is
preferable (Appendix~\ref{app:singlesk}).

\begin{figure}
	\centering
	\includegraphics[width=0.9\linewidth]{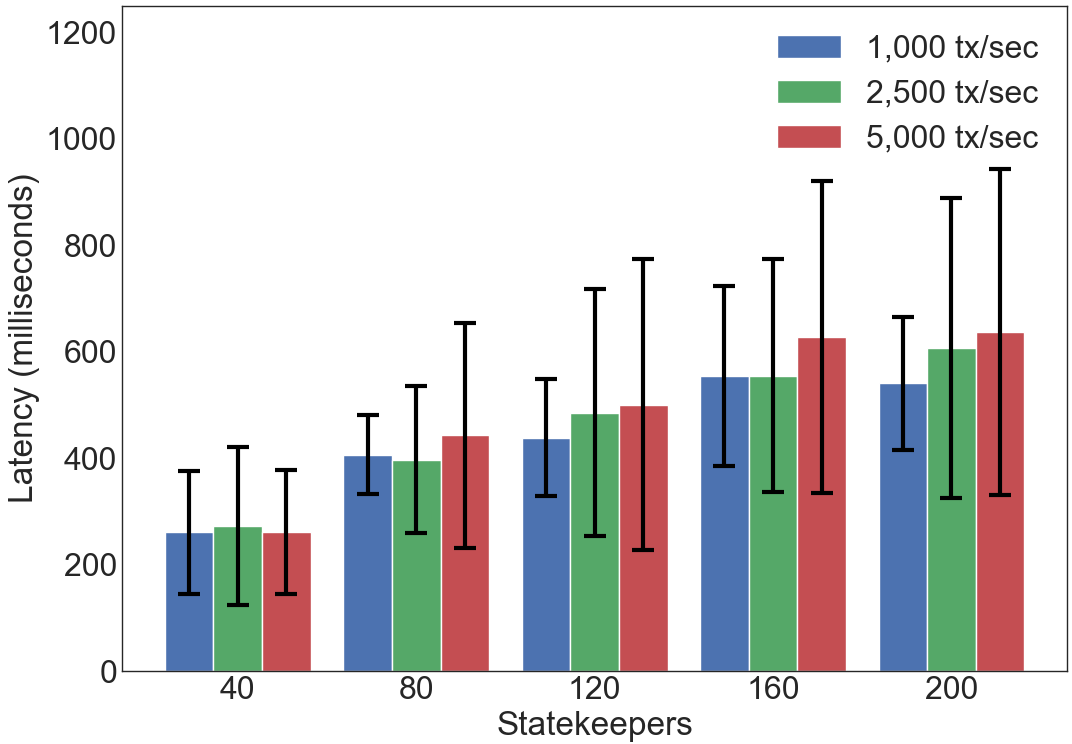}
	\caption{\textbf{Payment approval latency.}
	Payment approval latency for varying rates of incoming approval requests that each corresponds to one purchase.
	}
	\label{fig:tx}
	\figsaver
\end{figure}

\subsection{Scalability}\label{sec:scalability}
We now evaluate how many customers and merchants \sysname can handle (Requirement R2--large deployments).

Regarding the number of customers, the only scalability issue is that the recent transactions of each customer (e.g., past 24 hours) need to be recorded to the Arbiter's state. Thus, its storage needs grow linearly to the number of customers. Since Ethereum contracts do not have a storage limit, our implementation can facilitate hundreds of thousands or even millions of customers. In blockchain systems that allow smart contracts to examine past transactions, there is no need for the Arbiter to log transactions into the contract's state and the required storage is reduced significantly.

As shown in the previous section, \sysname can easily support 100-200 statekeeping merchants which is the 
intended usage scale of our solution. Moreover, due to the low communication complexity of our protocol, an optimized 
deployment could also support a few thousand statekeeping merchants with approval latency of less than 4 seconds.
However, in cases of very large deployments with several thousands merchants, it is preferable to allow
merchants to decide if they want to deposit a collateral and perform statekeeping tasks 
or simply be able to receive \sysname payments. Such a deployment remains trustless towards the 
statekeeping merchants and decentralized (but not fully), while it can support a much larger number of non-statekeeping merchants.
This design is further discussed in Appendix~\ref{app:dicentr}.

\subsection{Processing Cost}
\label{subsec:cost}

To evaluate our Requirement R3 (cheap payments), we implemented the Arbiter smart contract in Solidity for Ethereum, and measured the Ethereum gas cost of all \sysname operations. \update{6}{Our cost evaluation corresponds to a case where merchants run a \sysname instance non-profit. Joining a non-profit consortium allows merchants to accept fast and safe payments without having to pay fees to external entities such as card payment processors. Additional fees may be charged in for-profit setups.}
Table~\ref{tab:gascost} summarizes our results and provides the USD equivalents using the current conversion rate and a Gas price (Gwei) of 7.8.

\mypara{Registration cost.}
The one-time registration cost for merchants and customers is very low ($67,000$ Gas that equals to \$$0.06$), while statekeepers have to pay a slightly increased cost (\$$0.48$), due to verification of the proof of knowledge for the statekeeper's BLS private key to prevent \textit{rogue key} attacks~\cite{ristenpart2007power}.

The cost of the collateral clearance and withdrawal operations for both customers and statekeepers are also inexpensive, requiring \$$0.04$ and \$$0.02$.

\begin{table}[t]
	\centering
	\footnotesize
	\caption{Cost of \sysname operations.}
	\label{tab:gascost}
	\tabsaver
	\begin{tabular}{l r r}
		\toprule
		\textbf{Function} & \textbf{Gas} & \textbf{USD}\\
		\midrule
                Client/Merchant Registration & $67,000$ & $0.06$\\
                Statekeeper Registration & $510,000$ & $0.48$\\
                Clear Collateral & $42,000$ & $0.04$\\
                Withdraw Collateral& $23,000$ & $0.02$\\
                \textbf{Process Payment} & $169,000$ & \textbf{0.16} \\
                \bottomrule
	\end{tabular}
\end{table}

\mypara{Payment cost.}
The cost of a payment, in the absence of an attack, is $169,000$ Gas (\$$0.16$), mostly due to the cost of storing information about the transaction. This is roughly eight times as expensive as a normal Ethereum transaction that does not invoke a smart contract. In comparison, Mastercard service fees are ${\sim}1.5$\% of the transaction value~\cite{mastercardfees}. For example, a payment of \$$15$ will cost to the merchant $\$0.22$, while in a payment of $\$100$ the fees will rise to \$$1.5$.
\sysname compares favorably to these charges. BLS signatures~\cite{DBLP:journals/joc/BonehLS04} enabled us to aggregate the approval signatures, significantly reducing both the transaction size and the processing costs.

\mypara{Claim settlement cost.}
While our solution enables merchants to accept payments from customers with arbitrarily many pending transactions (constrained only by the customer's own collateral), the Ethereum VM and the block gas limits constrain the computations that can be performed during the settlement process. 
To examine these constraints, we consider the gas costs for different numbers of pending transactions, and statekeeper quorum sizes. While the number of the statekeepers is seemingly unrelated, it affects the cost of the settlement due to the aggregation of the public keys. In Appendix~\ref{app:generalizations}, we discuss a special case, where each customer is allowed to have only one transaction pending which simplifies settlement and reduces its cost significantly.

\begin{table}[t]
	\centering	
	\caption{Worst-case claim settlement cost in Gas and USD.}
	\label{tab:settlementcost}
	\tabsaver
	\resizebox{\columnwidth}{!}{
	\begin{tabular}{c c c c c}
		\toprule
		\multirow{2}{*}{\shortstack[c]{\textbf{Minimum}\\\textbf{Majority}}}
		& \multicolumn{4}{c}{\textbf{Pending Transactions per Customer}} \\\cline{2-5}
		& \sr 0 & 1 & 2 & 3 \\\hline
		\textbf{50}  & 1.9M (\$1.79) & 2.7M (\$2.54) & $3.5$M (\$$3.30$) & $4.3$M (\$$4.05$)\\
		\textbf{100} & 3.1M (\$2.92) & 4.3M (\$4.05) & $5.5$M (\$$5.19$) & $6.6$M (\$$6.22$)\\
		\textbf{150} & 4.2M (\$3.96) & 5.8M (\$5.47) & 7.4M (\$$6.98$) & $9.0$M (\$$8.49$)\\
		\textbf{200} & 5.4M (\$5.09) & 7.4M (\$6.984) & 9.4M (\$$8.87$) & $11.4$M (\$$10.75$)\\
		\textbf{250} & 6.6M (\$6.22) & 9.0M (\$8.494) & 11.3M (\$$10.66$) & $13.7$M (\$$12.93$)\\
		\bottomrule
	\end{tabular}
	}
\end{table}

Table~\ref{tab:settlementcost} shows the worst-case gas costs of settlement which is few dollars for a typical case (e.g., \$1.79 when $k=100$ and there is no pending transactions).\footnote{These costs were calculated by assuming that the adversary specifically crafts previous transactions to maximize the computational load of the settlement. To prevent an adversary from increasing the claim processing costs, the Arbiter contract could reject any transactions that have more approvals than the minimum necessary.}
Given Ethreum's gas limit of approximately $8$M gas units per block, \sysname claim settlement can either scale up to $k=499$ statekeepers ($250$ approvals) with no pending transactions ($6.6$M Gas), or to $3$ pending transactions per customer ($6.6$M Gas) with $k=199$ statekeepers ($100$ approvals). 

\subsection{Collateral Comparison}
\label{subsec:collaterals}

To evaluate our Requirement R2 (practical collaterals), we compare our deposits to known Layer-2 solutions, as shown in Table~\ref{tbl:opcost}. For our comparisons, we use example parameter values that are derived from real-life retail business cases. 

Our ``small shops'' scenario is based on sales numbers from~\cite{campion2011transit, pelham2010promoting, list2009economics}. The system has $n=100,000$ customers and $k=100$ merchants. The daily average expenditure of a customer per merchant is $e=\$10$ and the expenditure of a customer within $t=3$ minutes blockchain latency period is $e_t=\$5$. The number payments received by a merchant within the same time-period is $p_t=6$ (i.e., one customer payment every 30 seconds). 

Our ``large retailers'' example corresponds to the annual sales of a large retailer in the UK~\cite{tesco2018, tesco2016}. In this case, we have $n=1$ million customers, $k=100$ merchants, $e=\$250$, $e_t=\$100$ and $p_t=15$ (i.e., one payment every 12 seconds).

\begin{table}
	\centering
	\footnotesize
	\caption{Collateral comparison.}
	\label{tbl:opcost}
	\tabsaver
	\begin{tabular}{l c c c}
		\toprule
		\textbf{Solution} &\textbf{Customer} & \multicolumn{2}{c}{\textbf{Operator}} \\
				& & \textbf{Individual} & \textbf{Combined}\\ \hline\\[-0.25cm]
                Channels~\cite{miller2017sprites,LindNEKPS18} & $e \cdot k$  & & \\
                 \ $\hookrightarrow$ Small Shops & \$1,000  &  &  \\
                 \ $\hookrightarrow$ Large Retailers & \$25,000  &  & \\
				\hline\\[-0.25cm]
                Hubs~\cite{heilman2017tumblebit,dziembowski2019perun} & $e$ & \multicolumn{2}{c}{$\sum_n e$} \\
                 \ $\hookrightarrow$ Small Shops & \$10  & \multicolumn{2}{c}{\$1M} \\
                 \ $\hookrightarrow$ Large Retailers & \$250  & \multicolumn{2}{c}{\$250M} \\
                \hline \\[-0.25cm]
                
                \sysname & max($e_t$) & $\sum_k$ max($e_t$)$\cdot p_t$ & $\sum_k$ max($e_t$)$\cdot p_t\cdot k$ 	\\
                 \ $\hookrightarrow$ Small Shops & \$5  & \$3,000 & \$300,000 \\
                 \ $\hookrightarrow$ Large Retailers & \$100  & \$150,000 & \$15M \\
                \bottomrule
	\end{tabular}
\end{table}

\mypara{Customer collateral.}
In payment channels, the customer collateral grows linearly with the number of merchants. This leads to large customer collaterals ($\$1,000$ and $\$25,000$) in our example cases. Payment hubs alleviate this problem, but they still require customers to deposit their anticipated expenditure for a specific duration (e.g., $e=\$250$), and replenish it frequently. \sysname requires that customers deposit a collateral that is never spent (unless there is an attack) and equals the maximum value of payments they may conduct within the blockchain latency period (e.g., $e_t=\$100$).

\mypara{Merchant collateral.}
In payment hubs, the operator's deposit grows linearly with the number of customers in the system, since the operator must deposit funds equal to the sum of the customers' deposits~\cite{dziembowski2019perun,heilman2017tumblebit,gudgeonsok}. That is, the operator collateral must account for all registered customers, \emph{including the currently inactive ones}. Given our examples, this amounts to \$1M and \$250M for the ``small shops'' and the ``larger retailers'' cases, respectively. 

In \sysname, each merchant that operates as a statekeeper deposits enough funds to cover the total value of sales that the merchants conduct within the latency period $t$. Once a transaction gets successfully finalized on the blockchain (i.e., after $t$), the statekeeper can reuse that collateral in approvals of other payments. Thus, the total size of this collateral is independent of the number of registered customers, and is proportional to the volume of sales that merchants handle. In other words, the statekeeper collateral accounts only for the customers that are \emph{active} within the 3-minute latency period. Given our examples, this amounts to \$3,000 and \$150,000 which are three orders of magnitude less than in payment hubs. 
The combined deposit by all statekeepers (merchants) is shown on the last column of Table~\ref{tbl:opcost}. In both of our example cases, the combined deposit is smaller than in payment hubs. 

Designing a payment hub where the operator collateral is proportional to only active customers (and not all registered customers) is a non-trivial task, because the deposited collateral cannot be freely moved from one customer to another~\cite{khalil2018nc}. Some commit-chains variants~\cite{khalilnocust} manage to reduce operator collaterals compared to hubs, but such systems cannot provide secure \emph{and} fast payments. Other commit-chain variants~\cite{khalilnocust} enable fast and safe payments, but require online monitoring that is not feasible for all retail customers. Thus, we do not consider commit chains directly comparable and omit them here (see Section~\ref{sec:related_work} for more details).

\mypara{Cost of operation.} The amount of locked-in funds by system operators allows us to approximate the cost of operating \sysname system with respect to other solutions. For example, in our retail example case, each merchant that runs a statekeeper needs to deposit \$150k. Assuming 7\% annual return of investment, the loss of opportunity for the locked-in money is \$10,500 per year which, together with operational expenses like electricity and Internet, gives the operational cost of \sysname. We consider this an acceptable cost for large retailers. In comparison, a payment hub operator needs to deposit \$250M which means that running such a hub would cost \$17.5M plus operational expenses which is three orders of magnitude more.

\section{Discussion}
\label{sec:discussion}

\subsection{Governance} 
\label{sec:governance}

\update{2}{Snappy has, by design, a majority-based governance model. This means that any majority of statekeeping nodes can decide to ostracize merchants or statekeepers that are no longer deemed fit to participate. For example, a majority can stop processing requests from a merchant who has equivocated (signed conflicting transactions). If more complex governance processes are needed (e.g., first-past-the-post voting)
the Arbiter's smart contract can be extended accordingly.}

\subsection{Censorship Mitigation}
\label{subsec:censorship}

\update{3}{A direct implication of \sysname's governance model is that a majority of statekeeping nodes can discriminate against a particular victim merchant by not responding to its payment approval requests or by delaying the processing of its requests. Such censorship can be addressed in two ways.}

\update{3}{The first approach is technical. Targeted discrimination against a specific victim merchant could be prevented by hiding the recipient merchant's identity during the payment approval process. In particular, by replacing all the fields that could be used to identify the merchant (e.g., ``recipient'', ``amount'') from the payment Intent with a cryptographic commitment. Commitments conceal the merchant's identity from other merchants (statekeepers) during payment approval, but allow the merchant to claim any lost funds from the Arbiter later by opening the commitment. Moreover, if IP address fingerprinting is a concern, merchants can send their approval requests through an anonymity network (e.g., Tor would increase latency by ${\sim}$500ms~\cite{tormetrics2019}) or through the customer's device so as to eliminate direct communication between competing merchants.}

\update{3}{The second mitigation approach is non-technical. In case of known merchant identities and a mutually-trusted authority (e.g., a merchants' association), the victim merchant can file a complaint against constantly misbehaving merchants. In cases of widespread attacks, the victim merchant can reclaim their collaterals in full, deregister from this consortium and join another consortium.}

\subsection{Transaction Privacy}

\update{4}{For on-chain transaction privacy, \sysname inherits the privacy level of the underlying blockchain. For example, Ethereum provides transaction \emph{pseudonymity}, and thus every transaction that is processed with \sysname is pseudonymous once it recorded on the chain.}

\update{4}{During payment approval, the identity of the recipient merchant can be concealed from all the statekeepers using cryptographic commitments, as explained above (see Section~\ref{subsec:censorship}). However, the pseudonym of the customer remains visible to the statekeepers.} 

\update{4}{Well-known privacy-enhancing practices like multiple addresses and mixing services~\cite{ruffing2014coinshuffle, bonneau2014mixcoin} can be used to enhance customer privacy. For example, a \sysname customer could generate several Ethereum accounts, register them with the Arbiter and use each one of them only for a single payment. Once all accounts have been used, the customer can de-register them, generate a new set of accounts, move the money to the new accounts through a mixing service, and register new accounts. The main drawback of this approach is that the user needs to have more collateral locked-in and will pay the registration fee multiple times.}

\update{4}{In the future, privacy-preserving blockchains like ZCash~\cite{hopwood2016zcash} combined with private smart contracts based on Non-Interactive Zero-Knowledge proofs (NIZKs) could address the on-chain confidentiality problem more efficiently and protect the privacy of both the users and the merchants. However, realizing such a secure, efficient and private smartcontract language while achieving decent expressiveness, remains an open research problem~\cite{steffen2019zkay}.}

\subsection{Limitations}
\label{sec:limitations}

\update{8}{The main drawbacks of using \sysname are as follows. First, customers and merchants need to place small collaterals, and thus keep a percentage of their funds locked-in for extended periods of time. Second, \sysname can scale up to a moderate number of statekeeping merchants but cannot support hundreds of thousands or millions statekeeping nodes. In such cases, alternative deployment options can be used (see Appendix~\ref{app:dicentr}). Third, \sysname does not move the payment transactions off the chain and 
thus customers still need to cover the transaction processing fees charged by the blockchain's miners.}

\section{Related Work}
\label{sec:related_work}

\mypara{Payment channels} enable two parties to send funds to each other off the chain, while adding only an opening and a closing transaction on the chain~\cite{dziembowski2019perun,miller2017sprites,LindNEKPS18}. With the opening transaction the two parties lock funds in the channel, which are then used throughout the lifetime of the channel. 
In cases where the two parties send approximately the same amount of funds to each other over time, a payment channel can enable almost indefinite number of near-instant payments.
However, in the retail setting customers send funds unilaterally towards merchants. Moreover, customers transact with several merchants and thus each customer will need to maintain several channels and keep enough funds in them. 

\mypara{Payment networks} utilize the payment channels established between pairs of users to build longer paths~\cite{poon2016bitcoin,network2018cheap}. While this is a straightforward idea, in practice, finding routes reliably is not a trivial task~\cite{prihodko2016flare}. This is because the state of the individual channels changes arbitrarily over time and thus the capacity of the graph's edges fluctuate. Moreover, the unilateral nature of retail payments (customer $\rightarrow$ merchant) quickly depletes the available funds in the individual channels, preventing them from serving as intermediaries to route payments by other customers~\cite{engelmann2017towards}.
Miller et. al~\cite{miller2017sprites} showed that even under favorable conditions (2.000 nodes, customers replenish their accounts every $10$ seconds, maximum expenditure of $\$20$, no attacks), approximately $2\%$ of the payments will fail. At peak hours the ability of the network to route payments from customers to merchants is expected to degrade further. Rebalancing methods~\cite{khalil2017revive} have only a meager effect, primarily because credit cycles are rarely formed in real life~\cite{miller2017sprites}. 

\mypara{Payment hubs} introducing a single central point connecting all customers to all merchants. This eliminates the need of finding routing paths and in theory require a smaller total amount of locked funds for the customers~\cite{khalil2018nc}. However, this approach comes with two main drawbacks. First, it introduces a single point of failure for payment availability. And second, the hub operator needs to deposit very large amount of funds to match the total expenditure of all customers~\cite{dziembowski2019perun,heilman2017tumblebit} and thus will likely charge service fees. 
For instance, a hub that serves $n=1$M customers having in total of \$250M in their channels, must also lock-in that amount in channels with merchants to be able to accommodate payments, in particular during peak hours. Hub operators would charge significant fees to cover the opportunity cost of the large locked-in funds.

\mypara{Commit-chains} are parallel (and not yet peer-reviewed) work~\cite{khalil2018nc,khalilnocust,khalil2019system} that may either reduce or eliminate the operator collaterals compared to payment hubs. The main idea of commit-chains is to maintain a second-layer ledger and make periodic commitments (called \emph{checkpoints}) of its state transitions to the main chain. In one proposed variant~\cite{khalilnocust}, the central operator does not have to place any collateral, but such scheme does not enable fast and safe payments, because users need to wait for the next checkpoint which may take hours or days. Another proposed variant~\cite{khalilnocust} allows safe and fast payment, but has other problems. First, the users need to monitor the blockchain (hourly or daily) and dispute checkpoints if their balance is inaccurately represented. Such monitoring assumption is problematic, especially in use cases like retail with large number of customers using various client devices. Second, although the operator's collateral is slightly lower than those of payment hubs, it still remains very large (e.g., \$200M in our ``large retailers'' use case)~\cite{gudgeonsok}. \sysname enables fast and safe payments with smaller merchants collaterals for customers that remain mostly offline.

\mypara{Side-chains} use a permissioned set of validators to track pending transactions, typically using a BFT consensus protocol~\cite{dilley2016strong,back2014enabling}. 
Such solutions change the trust assumptions of permissionless blockchains significantly, as BFT consensus requires that 2/3 of the validators must be trusted. Side chains also require multiple rounds of communication and have high message complexity.

\mypara{Probabilistic payments} such as MICROPAY1/2/3 can in certain scenarios enable efficient and fast payment approval~\cite{salamon2018orchid, pass2015micropayments}. 
However, such solutions require that the service provided is continuous and granular so that the payments' probabilistic variance becomes negligible. In retail payments, this provides no guarantee that the merchant will be paid the right amount.

\section{Conclusion}
\label{sec:conclusion}

In this paper we have presented \sysname, a novel system that enables merchants to safely accept fast on-chain payments on slow blockchains. We have tailored our solution for settings such retail payments, where currently popular cryptocurrencies are not usable due to their high latency and previous solutions such as payment channels, networks and hubs have significant limitations that prevent their adoption in practice.

\section*{Acknowledgments}
The authors would like to thank the anonymous reviewers,
the shepherd Stefanie Roos, Mary Maller and George Danezis.
This research has been partially supported by the Zurich 
Information Security and Privacy Center (ZISC).

\bibliographystyle{IEEEtran}
\bibliography{references}

\normalsize
\appendix


\subsection{Background on Permissionless Consensus}
\label{app:consensus}

In this appendix we review additional examples of recent permissionless consensus system. We center our review around schemes that do not introduce significant additional security assumptions, and refer the reader to~\cite{bano2017consensus} for a more thorough survey.

\mypara{Proof of Stake} is proposed as an alternative to computational puzzles in Proof of Work,
where the nodes commit (i.e., stake) funds in order to participate in 
the consensus process~\cite{kwon2014tendermint,kiayias2017ouroboros,vasin2014blackcoin,li2017securing}. These solutions are based on an economic-driven model with nodes being 
rewarded for honest behavior and penalized for diverging from the consensus.
While they provide some latency improvements, they do not reach the $3$ seconds averages of
centralized payment processors. 

For example, Ouroboros~\cite{kiayias2017ouroboros} reports a throughput of approximately $300$ transactions per second and a block frequency 10-16 times smaller than that of Bitcoin. Thus, a merchant will still have to wait for several minutes before a transaction can be considered finalized. As another example of Proof-of-Stake system, Algorand~\cite{gilad2017algorand} has a throughput that is 125$\times$ higher than that of Bitcoin, and latency of at least $22$ seconds, assuming a network with no malicious users. 

\mypara{Sharding}. While Proof of Stake techniques involve the whole network in the consensus process,
sharding techniques promise significant performance improvement by splitting the 
network in smaller groups. For example, Elastico~\cite{luu2016secure} achieves four times larger throughput per epoch for network sizes similar to that of Bitcoin. However, it provides no improvement on the confirmation latency (${\sim800}$ seconds). Similarly, Gosig~\cite{li2018gosig} can sustain approximately $4.000$ transactions per second with an ${\sim}1$ minute confirmation time. Rapidchain~\cite{zamanirapidchain} has a throughput of $7,300$ transactions per second with a $70$-second confirmation latency. 

Omniledger~\cite{kokoris2018omniledger} is the only proposal that reports performance (both throughput and latency) compatible with retail payments. However, this comes at a security cost. Their low-latency transactions (i.e., Trust-but-Verify) use shards comprised of only a few (or even one) ``validators''. In the retail setting, this enables malicious validators to launch multi-spending attacks by approving several conflicting payments towards various honest merchants. While the attack and the malicious validator(s) will be uncovered by the core validators within ${\sim}1$ minute, this time-period suffices for an adversary to attack multiple merchants resulting in substantial losses. 

Shasper for Ethereum is expected to be capable of handling $13,000$ transactions per second, while 
optimistic estimates report a minumum block frequency of ${\sim}8$ seconds ~\cite{shasper1, shasper2,buterin2014slasher,buterin2017casper}. However, even with such a small block interval, the latency remains too high for retail purchases, as merchants will need to wait for several blocks for transactions to eventually reach finality.

\subsection{Deployment Alternatives for \sysname}
\label{app:generalizations}

In this appendix we discuss alternative deployment options.

\mypara{Centralized \sysname.}
\label{app:singlesk}
While decentralization is one of the main benefits of
\sysname, there may be cases where having one
central party is also acceptable. In this case, 
the merchants can appoint a single party
to approve or reject payments.
This simplifies our protocols as
each payment requires only one 
approval query instead of several.
However, if the central party is untrusted, 
it still has to deposit a collateral
for the merchants to claim in case it equivocates.
However, relying on a single party comes with drawbacks.
In particular, such a setup adds a centralized layer on top of a fully
decentralized blockchain, while the reliance on a single service provider 
will likely result in increased service fees. 

\begin{figure}
	\centering
	\includegraphics[width=0.5\linewidth]{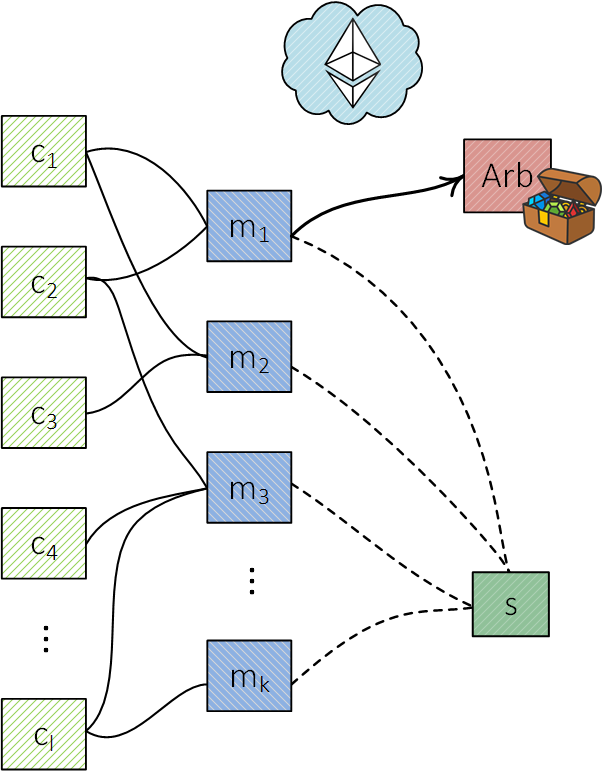}
	\caption{\textbf{Centralized \sysname instance.} The customers $c_i$ initiate payments towards the merchants $m_j$, who then consult with the central statekeeper $s$ and either accept or reject the payment.}
	\label{fig:utp}
	\centering
	\figsaver
\end{figure}

\mypara{Non-statekeeping Merchants}\label{app:dicentr}
In the previous section, we discussed a fully centralized version of \sysname that
allows the system to scale even further and simplifies statekeeping.
While this setup has several advantages, 
the liveness and quality of service of the deployment relies on the single party that
can unilaterally decide on changing the service fees and processes.
An alternative solution that retains most of the centralization benefits
while remaining decentralized (but not fully) is allowing non-statekeeping merchants.
Merchants can join the system and decide if they want to 
deposit a collateral and perform statekeeping tasks or they simply want to be able to receive
payments.

This allows merchants who choose not to keep state to still make use of the \sysname deployment
and enables the system to scale to millions of merchants.
To incentivize statekeeping, a small service fee could be paid by non-statekeeping merchants.
to those who have allocated a statekeeping collateral.
While this goes beyond the scope of this paper, it is probably preferable
if the statekeepers' set remains open to all interested merchants (cf. to being capped or fixed)
and the service fees are determined dynamically based on the offer/demand.
Note that several merchants may be represented by a single statekeeping node.
For example, instead of having small shops match the collateral of large chain stores,
their association could maintain one node that performs the statekeeping for all of them
and routes their to-be-approved payments to the rest of the statekeepers.

This setup has the advantage of being able to use any trust relationships 
(e.g., small merchants trust their association) when/where they exist,
while still allowing a trustless setup for actors who prefer it.

\mypara{One pending transaction.}
\label{app:singlepend}
Much of the complexity of \sysname's protocols comes from the 
fact that the pending transactions of a customer should never exceed in
value the collateral. One possible way to reduce this complexity is
by constraining the number of allowed pending transactions to $1$.
Such a setup allows the customers to conduct at most one
transaction per block, and greatly simplifies the settlement process,
as there are no pending transaction to be provided by the merchant.
We believe that such a setup is realistic and may be preferable
in cases where the customers are unlikely to perform several
transactions within a short period of time. However, this is
an additional assumption that may reduce the utility of the system
in some cases. For example, a customer, who after checking out
realizes that they forgot to buy an item, will have to wait until 
the pending transaction is confirmed.

\mypara{Signature verification batching.}
While the computational cost of verifying an aggregated signature (i.e., two pairings)
is negligible for a personal computer, this is not true for the
Ethereum Virtual Machine, where a pairing operation is considerably more expensive than
a group operation. Our original scheme tackles this cost by having the arbiter verify signatures only in case of disputes. As an additional cost-reduction optimization, the arbiter can use techniques such as those in~\cite{DBLP:conf/eurocrypt/BellareGR98} to batch and check several signatures simultaneously

Let's assume that there are $\ell$ aggregated signatures $(\sigma_1, \ldots, \sigma_\ell)$ to be verified for the messages $(m_1, \ldots , m_\ell)$.
The arbiter samples $\ell$ random field elements $(\gamma_1, \ldots, \gamma_\ell)$ from $\Z_p$.
The verifier considers all the signatures to be valid if
$$e\left(\prod_{i=1}^{\ell}\sigma_i^{\gamma_i},h \right) = \prod_{i=1}^{\ell} e\left(H(m_i)^{\gamma_i}, \prod_{j =1, \tau_{q_i}[j] = 1}^{n}  v_j\right).$$
This roughly halfs the verification costs.
In systems where the number of transactions is considerably more than the number of statekeepers, we can reduce the costs per transaction further.
Assume that there are $\ell$ aggregated signatures $(\sigma_1, \ldots, \sigma_\ell)$ to be verified for the messages $(m_1, \ldots , m_\ell)$ where $\ell \gg n$.
The verifier samples $\ell$ random field elements $(\gamma_1, \ldots, \gamma_\ell)$ from $\Z_p$.
The verifier considers all the signatures to be valid if
$$e\left(\prod_{i=1}^{\ell}\sigma_i^{\gamma_i},h\right) = \prod_{j=1}^{n} e\left(\prod_{i=1, \tau_{q_i}[j]=1 }^{\ell} H(m_i)^{\gamma_i}, v_j\right).$$
The cost of verifying a batch of $\ell$ signatures signed by $n$ merchants is then
$n + 1$ pairing operations and $2 \ell$ group exponentiations in $\Gr$.

\mypara{Dynamic Statekeepers' Consortia.}
So far we have considered only cases where customers joined and withdrew from the system. 
Similarly, one can imagine that statekeepers could decide to leave or new statekeepers may want to join
an existing deployment.
Such functionality can be easily facilitated by modifying through the registration and de-registration 
algorithms available for customers. However, while churning customers do not pose a threat to the security 
of the system, changes in the set of statekeepers may result in attacks.

\begin{figure}[]
	\centering
	\includegraphics[width=0.6\linewidth]{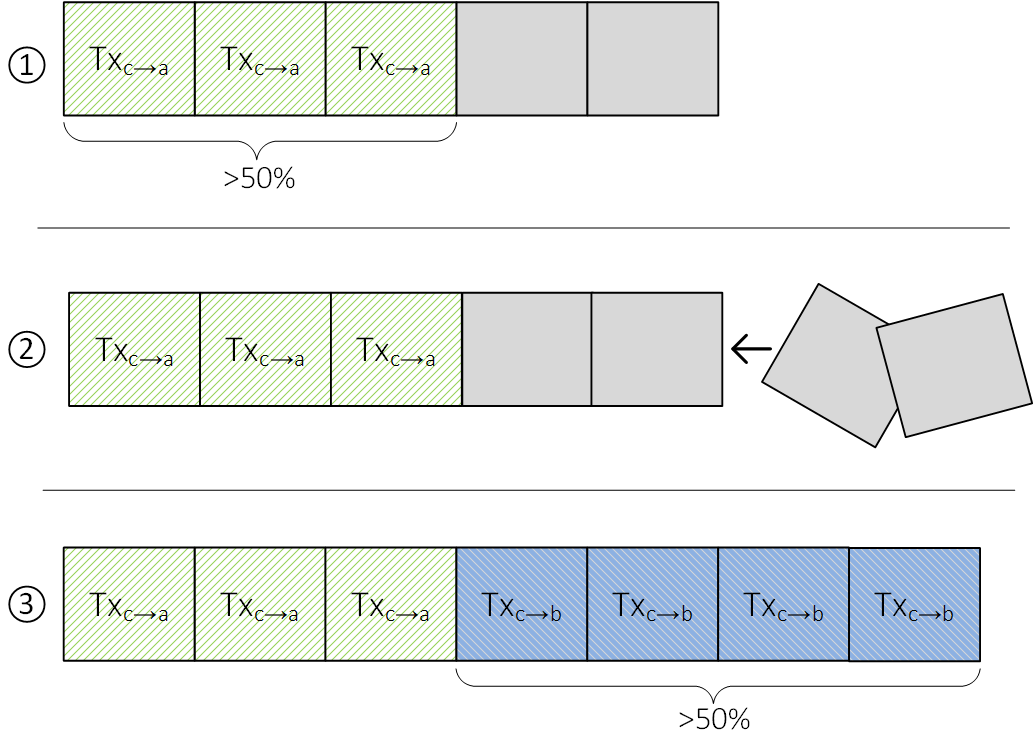}
	\caption{\textbf{Moving majority attack} enables a customer to have two transactions with the same index value approved by disjoint statekeeper majorities.}
	\label{fig:movingmajority}
	\figsaver
\end{figure}

Such an attack could enable a malicious customer to have two approved transactions with the same
index value. As shown in Figure~\ref{fig:movingmajority}, initially the system 
features $5$ statekeepers ($s_1 \ldots s_5$) and a merchant who wants to have 
a transaction $\tx{}$ approved, reaches out to $s_1$, $s_2$ and $s_3$.
Subsequently, two new statekeepers $s_6$ and $s_7$ join the system.
The malicious customer now issues another transaction $\tx{}'$, such that $\tx{i}=\tx{i}'$.
The merchant receiving $\tx{}'$ now queries a majority of the statekeepers (i.e., $s_4$, $s_5$,
$s_6$ and $s_7$) and gets the transaction approved.
At the final stage of the attack $c$ issues another transaction that invalidates
$\tx{}$ and $\tx{}'$ (e.g., a doublespend), while the malicious merchant
quickly claims $\tx{v}'$ from the customer's collateral. 
Because of the way the arbiter processes claims (Line~\ref{line:not_processed2} in Algorithm~\ref{alg:claim_cust}), the honest merchant is now unable to claim $\tx{v}$ from $\col{c}$.
Moreover, none of the statekeepers equivocated and thus no funds can be recouped from their
collateral.

\sysname can safely support a dynamically changing set of statekeepers, if appropriate protection mechanisms are deployed.
For example, such attacks can be prevented:
1) by giving early notice to the merchants about changes in the set (e.g., multistage registration), and
2) by waiting (a few minutes) until all past transactions are finalized in the blockchain before/after every change.

\end{document}